\documentclass[12pt]{article}
\usepackage{amsmath}
\usepackage{graphicx,psfrag,epsf}
\usepackage{enumerate}
\usepackage{psfrag,epsf}
\usepackage{url} 
\usepackage{amsmath,amssymb,amsfonts,amsthm,mathtools,mathrsfs}
\usepackage{bm,bbm}
\usepackage{color}
\usepackage{subfigure}
\usepackage{algorithm,algpseudocode}
\usepackage[symbol]{footmisc}
\usepackage{scalefnt}
\usepackage{authblk} 
\usepackage{multirow}
\usepackage[colorlinks=true, citecolor=blue, urlcolor=blue]{hyperref}
\usepackage{natbib}
\usepackage{dsfont}
\usepackage[title]{appendix}
\usepackage{newtxtext,newtxmath} 
\usepackage{mhchem} 

\input cyracc.def

\usepackage[left=1in,top=1.1in,right=0.5in,bottom=1in]{geometry}

\usepackage{letltxmacro}
\makeatletter
\let\oldr@@t\r@@t
\def\r@@t#1#2{%
	\setbox0=\hbox{$\oldr@@t#1{#2\,}$}\dimen0=\ht0
	\advance\dimen0-0.2\ht0
	\setbox2=\hbox{\vrule height\ht0 depth -\dimen0}%
	{\box0\lower0.4pt\box2}}
\LetLtxMacro{\oldsqrt}{\sqrt}
\renewcommand*{\sqrt}[2][\ ]{\oldsqrt[#1]{#2}}
\makeatother
\theoremstyle{definition}

\newtheorem{theorem}{Theorem}[section]

\newtheorem{corollary}[theorem]{Corollary}

\newtheorem{definition}{Definition}[section]
\newtheorem{lemma}[theorem]{Lemma}
\newtheorem{proposition}[theorem]{Proposition}
\newtheorem{remark}[theorem]{Remark}
\numberwithin{equation}{section} 

\makeatletter
\def\@seccntformat#1{\@ifundefined{#1@cntformat}%
	{\csname the#1\endcsname\quad}
	{\csname #1@cntformat\endcsname}
}
\makeatother

\newif\ifShowComments
\ShowCommentstrue
\def\strutdepth{\dp\strutbox}
\def\druk#1{\strut\vadjust{\kern-\strutdepth
        {\vtop to \strutdepth{%
                \baselineskip\strutdepth\vss
                        \llap{\hbox{#1}\quad}\null}}}}

\title{\bf
Bivariate distributions on the unit square: Theoretical properties and applications
}

\author[1, 2]{  Roberto Vila \thanks{rovig161@gmail.com}}
\author[2]{Narayanaswamy Balakrishnan  \thanks{bala@mcmaster.ca
} }
\author[1]{Helton Saulo  \thanks{heltonsaulo@gmail.com} }
\author[1]{ \\ Peter Zörnig \thanks{peter@unb.br}}
\affil[1]{Department of Statistics, University of
	Bras\'ilia, Bras\'ilia, Brazil}
\affil[2]{
	Department of Mathematics and Statistics, McMaster University, Hamilton, Ontario, Canada}
\begin{document}
	\maketitle 	
	\begin{abstract}
		{
We introduce the bivariate unit-log-symmetric model based on the bivariate log-symmetric distribution (BLS) defined in \cite{Vila2022} as a flexible family of bivariate distributions over the unit square. We then study its mathematical properties such as stochastic representations, quantiles, conditional distributions, independence of the marginal distributions and moments. Maximum likelihood estimation method is discussed and examined through Monte Carlo simulation. Finally, the proposed model is used to analyze soccer data.
		}

	\end{abstract}
	
	
	\section{Introduction}
	\noindent
	
Bivariate distributions over the unit-square have been discussed in detail in the literature. Many of them are based on beta distribution and its generalizations; see \cite{ArnoldNg2011} and \cite{Nadarajahetal2017}. Models of this type have been studied since the 1980s. Some other distributions on the unit square are based on generalized arcsine and inverse Gaussian distributions. A recent model, the bivariate unit-sinh-normal distribution, is based on the bivariate Birnbaum-Saunders distribution; see \cite{Martinez2022}. Bivariate distributions over the unit square arise naturally in comparing indices, rates or proportions in the interval $(0, 1)$.
	
In this paper, we study the bivariate unit-log-symmetric (BULS) distribution defined over the unit-square, obtained as a modification of the bivariate log-symmetric (BLS) distribution introduced by \cite{Vila2022}. The definitions of BLS and BULS distributions are given in Section \ref{Sec:2}, along with some special cases of BULS. In Section \ref{Sec:3}, we discuss some properties of the new model, including a stochastic representation, marginal quantiles, and the conditional distributions of BULS. We derive more compact formulas for the conditional densities, using the distribution functions of normal, Student-$t$, hyperbolic, Laplace and slash distributions.
One of the uses of having  closed formulas for the conditional densities (of the BULS model), for example, is in studying Heckman-type selection models \citep{Heckman1979} when the selection variables have bounded support. In addition, we derive the distribution of the squared Mahalanobis distance of a random vector $\boldsymbol{W}=(W_1,W_2)^\top$ with BULS distribution, and present a necessary condition for the independence of the components of $\boldsymbol{W}$ and formulas for the moments of $W_1$ and $W_2$. In Section \ref{Sec:5}, the log-likelihood function and the likelihood equations for the BULS distribution are presented. In Section \ref{Sec:4}, we carry out a Monte Carlo simulation study to evaluate the performance of the ML estimators by means of their bias, root mean square error and coverage probability. In Section \ref{Sec:6}, we present two applications to soccer data. Specifically, in Section \ref{Sec:6.1}, we model the vector $\boldsymbol{W}=(W_1,W_2)^\top$, where $W_1$ represents the time elapsed until a first kick goal (of any team) and the time elapsed until a goal of any type of the home team, and show that specific BULS distributions are suitable for modelling $\boldsymbol{W}$. In Section \ref{Sec:6.2}, we consider the data of 2022 FIFA World Cup wherein the components of the vector $\boldsymbol{W}$ represent the pass completion proportions of medium passes (14 to 18 meters) and long passes (longer than 37 meters). We then demonstrate that these data can also be fitted well by BULS distributions.


	\section{Bivariate unit-log-symmetric model}\label{Sec:2}
	\noindent
	In this section, we describe the bivariate unit-log-symmetric model (BULS). To define this model, we first need to describe the bivariate log-symmetric distribution (BLS) defined in \cite{Vila2022}.
	
	\subsection{BLS distribution}
	Following \cite{Vila2022}, a continuous random vector $\boldsymbol{T}=(T_1,T_2)^\top$ is said to have a bivariate log-symmetric (BLS) distribution if its joint probability density function (PDF) is given by
	\begin{eqnarray}\label{PDF}
	f_{T_1,T_2}(t_1,t_2;\boldsymbol{\theta})
	=
	{1\over t_1t_2\sigma_1\sigma_2\sqrt{1-\rho^2}Z_{g_c}}\,
	g_c\Biggl(
	{\widetilde{t_1}^2-2\rho\widetilde{t_1}\widetilde{t_2}+\widetilde{t_2}^2
		\over 
		1-\rho^2}
	\Biggr),
	\quad 
	t_1, t_2>0,
\end{eqnarray}
where
$\widetilde{t_i}
	=
	\log[({t_i/\eta_i})^{1/\sigma_i}], \ \eta_i=\exp(\mu_i), \ i=1,2$,
	with $\boldsymbol{\theta}=(\eta_1,\eta_2,\sigma_1,\sigma_2,\rho)^{\top}$ being the parameter vector, $\mu_i\in\mathbb{R}$, $\sigma_i>0$, $i=1,2$ and $\rho\in(-1,1)$. 
	Furthermore, $Z_{g_c}>0$ is the partition function, that is,
	\begin{align}\label{partition function}
	Z_{g_c}
	&=
	\int_{0}^{\infty}\int_{0}^{\infty}
	{1\over t_1t_2\sigma_1\sigma_2\sqrt{1-\rho^2}}\,
	g_c\Biggl(
	{\widetilde{t_1}^2-2\rho\widetilde{t_1}\widetilde{t_2}+\widetilde{t_2}^2
		\over 
		1-\rho^2}
	\Biggr)\, {\rm d}t_1{\rm d}t_2
	=
	\pi \int_{0}^{\infty}  g_c(u)\,{\rm d}u,
	\end{align}
	and $g_c$ is a scalar function referred to as the density generator \cite[see][]{Fang1990}. The second integral in \eqref{partition function} is consequence of a change of variables; for more details, see Proposition 3.1 of \cite{Vila2022}.
	When a random vector $\boldsymbol{T}$ is BLS distributed, with parameter vector $\boldsymbol{\theta}$, we denote it by $\boldsymbol{T}\sim {\rm BLS}(\boldsymbol{\theta},g_c)$.

	\subsection{BULS distribution}\label{The BULS distribution}
	We say that a continuous random vector $\boldsymbol{W}=(W_1,W_2)^\top$ has a bivariate unit-log-symmetric (BULS) distribution with parameter vector $\boldsymbol{\theta}=(\eta_1,\eta_2,\sigma_1,\sigma_2,\rho)^{\top}$, denoted by $\boldsymbol{W}\sim{\rm BULS}(\boldsymbol{\theta},g_c)$, if its PDF is, for $0<w_1, w_2<1$, given by
	\begin{eqnarray}\label{PDF-BULS}
	\!\!f_{W_1,W_2}(w_1,w_2;\boldsymbol{\theta})
	=
	{1\over 
		(1-w_1)t_1\sigma_1 (1-w_2)t_2\sigma_2 \sqrt{1-\rho^2}Z_{g_c}}\,
	{\displaystyle 
		g_c\Biggl(
		{\widetilde{w}_1^2-2\rho\widetilde{w}_1\widetilde{w}_2+\widetilde{w}_2^2
			\over 
			1-\rho^2}
		\Biggr)
	},
\end{eqnarray}
where $\widetilde{w}_i
	=
	\log[({t_i/ \eta_i})^{1/\sigma_i}],  
	\ t_i=-\log(1-w_i), \ \eta_i=\exp(\mu_i), \ i=1,2$,
	with 
	$\sigma_i>0$, $i=1,2$, $\rho\in(-1,1)$, and $Z_{g_c}$ and ${g_c}$ are as given in \eqref{partition function}. We shall prove later that the BULS PDF in \eqref{PDF-BULS} is obtained by taking $W_i=1-\exp(-T_i)$, $i=1,2$, with $=(T_1,T_2)^{\top}\sim {\rm BLS}(\boldsymbol{\theta},g_c)$.

	
	Table \ref{table:1} presents some examples of bivariate unit-log-symmetric distributions.
	\begin{table}[H]
		\caption{Partition functions $(Z_{g_c})$ and density generators $(g_c)$ for some BULS distributions.}
		\vspace*{0.15cm}
		\centering 
			\begin{tabular}{llll} 
				\hline
				Distribution 
				& $Z_{g_c}$ & $g_c$ & Parameter 
				\\ [0.5ex] 
				\noalign{\hrule }
				Bivariate unit-log-normal
				& $2\pi$ & $\exp(-x/2)$ & $-$ 
				\\ [1ex] 
				Bivariate unit-log-Student-$t$
				& ${{\Gamma({\nu/ 2})}\nu\pi\over{\Gamma({(\nu+2)/ 2})}}$  
				& $(1+{x\over\nu})^{-(\nu+2)/ 2}$  &  $\nu>0$
				\\ [1ex]
				Bivariate unit-log-hyperbolic
				& ${2\pi (\nu+1)\exp(-\nu)\over \nu^2}$ & $\exp(-\nu\sqrt{1+x})$ &  $\nu>0$
				\\ [1ex]   
				Bivariate unit-log-Laplace
				& $\pi$  & $K_0(\sqrt{2x})$ & $-$
				\\ [1ex]   
				Bivariate unit-log-slash
				& ${\pi\over q}\, 2^{2-q\over 2}$ & $ x^{-{q+2\over 2}} \gamma({q+2\over 2},{x\over 2})$ & $q>0$
				\\ [1ex]   
				\hline	
		\end{tabular}
		\label{table:1} 
	\end{table}
	\noindent
	In Table \ref{table:1}, $\Gamma(t)=\int_0^\infty x^{t-1} \exp(-x) \,{\rm d}x$, $t>0$, is the complete gamma function,
	$K_\lambda(u)=(1/2)(u/2)^\lambda\int_0^\infty t^{-\lambda-1} \exp(-t-{u^2/ 4t}) \,{\rm d}t$, $u>0$, is the 
	modified Bessel function
	of the third kind with index $\lambda$
	\cite[see Appendix of ][]{Kotz2001},
	and $\gamma(s,x)=\int_{0}^{x}t^{s-1}\exp(-t)\,{\rm d}t$ is the lower incomplete gamma function. 
	%
	
	Let $\boldsymbol{T}=(T_1,T_2)^{\top}\sim {\rm BLS}(\boldsymbol{\theta},g_c)$.
	From \eqref{PDF-BULS}, it is clear that the random vector $\boldsymbol{X}=(X_1,X_2)^{\top}$, with
	\begin{align}\label{id-X-W}
	X_i=\log(T_i)=\log\big[-\log(1-W_i)\big], 
	\quad i=1,2,
	\end{align}
	has a bivariate elliptically symmetric (BSY) distribution  \cite[see p. 592 in][]{Bala2009}; that is, the PDF of $\boldsymbol{X}$ is
	\begin{eqnarray}\label{PDF-symmetric}
	f_{X_1,X_2}(x_1,x_2;\boldsymbol{\theta}_*)
	=
	{1\over \sigma_1\sigma_2\sqrt{1-\rho^2}Z_{g_c}}\,
	g_c\Biggl(
	{\widetilde{x_1}^2-2\rho\widetilde{x_1}\widetilde{x_2}+\widetilde{x_2}^2
		\over 
		1-\rho^2}
	\Biggr),
	\quad 
	-\infty<x_1,x_2<\infty,
\end{eqnarray}
where $\widetilde{x_i}={(x_i-\mu_i)/\sigma_i}, \ i=1,2$,
	with $\boldsymbol{\theta}_*=(\mu_1,\mu_2,\sigma_1,\sigma_2,\rho)$ being the parameter vector and $Z_{g_c}$ is the partition function defined in \eqref{partition function}. In this case, we shall use the notation $\boldsymbol{X}\sim {\rm BSY}(\boldsymbol{\theta}_*, g_c)$.
	
	It is a simple task to observe that the joint cumulative distribution function (CDF) of $\boldsymbol{W}\sim {\rm BULS}(\boldsymbol{\theta},g_c)$, denoted by $F_{W_1,W_2}(w_1,w_2;\boldsymbol{\theta})$, is given by
	\begin{align*}
	F_{W_1,W_2}(w_1,w_2;\boldsymbol{\theta})
	&=
	F_{T_1,T_2}\big(-\log(1-w_1), -\log(1-w_2);\boldsymbol{\theta}\big)
	\\[0,2cm]
	&=
	F_{X_1,X_2}\big(\log[-\log(1-w_1)],\log[-\log(1-w_2)];\boldsymbol{\theta}_*\big),
	\end{align*}
	wherein $F_{T_1,T_2}(t_1,t_2;\boldsymbol{\theta})$ and $F_{X_1,X_2}(x_1,x_2;\boldsymbol{\theta}_*)$ denote the CDFs of $\boldsymbol{T}\sim {\rm BLS}(\boldsymbol{\theta},g_c)$ and $\boldsymbol{X}\sim {\rm BES}(\boldsymbol{\theta}_*, g_c)$, respectively. Note that there is no closed form expression for the CDF of $\boldsymbol{X}$ with the exception of bivariate normal.


	\section{Some basic properties of the model} \label{Sec:3}
	\noindent
	
	In this section, some mathematical properties of the bivariate unit-log-symmetric distribution are established.
	
	\subsection{Stochastic representation}\label{st-rep}
	
	\begin{proposition}\label{Stochastic Representation}
		The random vector $\boldsymbol{W}=(W_1,W_2)^\top$ has a BULS distribution if
		\begin{align*}
		\begin{array}{lllll}
		&W_1=1-\exp\big[-\eta_1 \exp(\sigma_1 Z_1)\big],
		\\[0,4cm]
		&W_2=
		1-\exp\big[-\eta_2 
		\exp\big(\sigma_2 {\rho} Z_1+\sigma_2\sqrt{1-\rho^2} Z_2\big)\big],
		\end{array}
		\end{align*} 
		where $Z_1=RDU_1$ and $Z_2=R\sqrt{1-D^2}U_2$ with
		$U_1$, $U_2$, $R$, and $D$ being mutually independent random variables, $\rho\in(-1,1)$, $\eta_i=\exp(\mu_i)$, and $\mathbb{P}(U_i = -1) = \mathbb{P}(U_i = 1) = 1/2$, $i=1,2$. The random variable $D$ is positive  and has PDF
		$
		f_D(d)={2/(\pi\sqrt{1-d^2})}, \ d\in(0,1).
		$
		Further, the positive random variable $R$ 
		has its PDF as
		$
		f_R(r)={2r g_c(r^2)/\int_{0}^{\infty}
			g_c(u)
			\, {\rm d}{u}}, \ r>0.
		$
	\end{proposition}
	\begin{proof}
		It is well-known that \cite[see Proposition 3.2 of][]{Vila2022} the random vector $\boldsymbol{T}=(T_1,T_2)^\top$ has a BLS distribution if
		\begin{align}\label{rep-stoch-biv-gaussian}
		\begin{array}{lllll}
		&T_1=\eta_1 \exp(\sigma_1 Z_1),
		\\[0,4cm]
		&T_2=\eta_2 
		\exp\big(\sigma_2 {\rho} Z_1+\sigma_2\sqrt{1-\rho^2} Z_2\big).
		\end{array}
		\end{align} 
		Moreover, from \eqref{id-X-W}, $W_i=1-\exp(-T_i)$, $i=1,2$. Hence, the result.
	\end{proof}
	

The following lemma provides a slight simplification in the representation of Proposition \ref{Stochastic Representation}. This result plays a fundamental role in the next subsections, since all the probabilistic characteristics that depend on the distribution of $\rho Z_1+\sqrt{1-\rho^2}Z_2$ will be simplified since it has the same distribution as $Z_2$.

\begin{lemma}\label{Prop-dual-1}
	For a Borelian subset $B$ of $(0,1)$, we have
	\begin{align*}
	\mathbb{P}\left(\rho Z_1+\sqrt{1-\rho^2}Z_2\in {B}\right)
	=
	\mathbb{P}(Z_2\in B).
	\end{align*}
	In other words, $\rho Z_1+\sqrt{1-\rho^2}Z_2$ and $Z_2$  have the same distribution.
\end{lemma}
\begin{proof}
	It is clear that the density of $\rho Z_{1}+\sqrt{1-\rho^2}{Z_{2}}$ is related to the joint density $f_{Z_{1}, Z_{2}}$ by
	\begin{align}\label{int-density-sum}
	f_{\rho Z_{1}+\sqrt{1-\rho^2}{Z_{2}}}(s_2)
	=
	{1\over \sqrt{1-\rho^2}}
	\int_{-\infty}^\infty
	f_{Z_{1}, Z_{2}}\left(z, \dfrac{s_2-\rho z}{\sqrt{1-\rho^2}}\right) {\rm d}z.
	\end{align}
	From Eq. (13) of \cite{Saulo2022}, the joint PDF of $Z_1$ and $Z_2$ is given by
	\begin{align}\label{pdf-Zs}
	f_{Z_1,Z_2}(x,y)={1\over Z_{g_c}}\, g_c(x^2+y^2), \quad -\infty<x,y<\infty,
	\end{align}
	and so the integral in \eqref{int-density-sum} is
	\begin{align}\label{int-H}
	=
	{1\over \sqrt{1-\rho^2}Z_{g_c}}
	\int_{-\infty}^\infty
	g_c\left(z^2 +\bigg(\dfrac{s_2-\rho z}{\sqrt{1-\rho^2}}\bigg)^2\right)  
	{\rm d}z.
	\end{align}
	Using the identity
	\begin{align*}
	z^2 +\bigg(\dfrac{s_2-\rho z}{\sqrt{1-\rho^2}}\bigg)^2
	=
	\dfrac{z^2-2\rho zs_2+s_2^2}{1-\rho^2}
	=
	\bigg(\dfrac{z-\rho s_2}{\sqrt{1-\rho^2}}\bigg)^2+s_2^2
	\end{align*}
	the integral in \eqref{int-H} is written as 
	\begin{align*}
	=
	{1\over \sqrt{1-\rho^2}Z_{g_c}}
	\int_{-\infty}^\infty
	g_c\left(\bigg(\dfrac{z-\rho s_2}{\sqrt{1-\rho^2}}\bigg)^2+s_2^2\right) 
	{\rm d}z.
	\end{align*}
	Making the change of variables
	$s_1={(z-\rho s_2)}/{\sqrt{1-\rho^2}}$,
	the above integral is
	\begin{align*}
	=
	{1\over Z_{g_c}}\,
	\int_{-\infty}^\infty
	g_c (s_1^2+s_2^2)
	\, {\rm d}s_1
	=
	\int_{-\infty}^\infty
	f_{Z_1,Z_2}(s_1,s_2)\, {\rm d}s_1,
	\end{align*}
	where, in the last line, we have used \eqref{pdf-Zs}. Hence,
	\begin{align}\label{identity-sum}
	f_{\rho Z_{1}+\sqrt{1-\rho^2}{Z_{2}}}(s_2)
	=
	\int_{-\infty}^\infty
	f_{Z_1,Z_2}(s_1,s_2)\, {\rm d}s_1
	=
	f_{Z_2}(s_2).
	\end{align}
	
	Now, from \eqref{identity-sum}, it is clear that  $\rho Z_1+\sqrt{1-\rho^2}Z_2$ and $Z_2$ are equal in distribution.
\end{proof}

	\subsection{Marginal Quantiles}\label{marg-quant}
	
	Given $p\in(0,1)$, let $Q_{W_i}(p)$ be the $p$-quantile  of $W_i$, for $i=1,2$.
	By using the stochastic representation in Proposition \ref{Stochastic Representation}, for $\boldsymbol{W}=(W_1,W_2)^\top\sim {\rm BULS}(\boldsymbol{\theta},g_c)$, we have
	\begin{align*}
	p=
	\mathbb{P}(W_1\leqslant Q_{W_1}(p))
	&=
	\mathbb{P}\big(1-\exp\big[-\eta_1 \exp(\sigma_1 Z_1)\big]\leqslant Q_{W_1}(p)\big)
	\\[0,2cm]
	&=
	\mathbb{P}\left(Z_1\leqslant \log\biggl[\left(-{ \log(1-Q_{W_1}(p))\over \eta_1}\right)^{1/\sigma_1}\biggl]\right)
	\end{align*}
	and 
	\begin{align*}
	p=\mathbb{P}(W_2\leqslant Q_{W_2}(p))
	&=
	\mathbb{P}\left(
	1-\exp\left[-\eta_2 
	\exp\left(\sigma_2 {\rho} Z_1+\sigma_2\sqrt{1-\rho^2} Z_2\right)\right]
	\leqslant Q_{W_2}(p)
	\right) 
	\\[0,2cm]
	&=
	\mathbb{P}\left(
	{\rho} Z_1+\sqrt{1-\rho^2} Z_2
	\leqslant 
	\log\biggl[\left(-{\log(1-Q_{W_2}(p))\over \eta_2}\right)^{1/\sigma_2}\biggl]\right).
	\end{align*}
	Hence, the $p$-quantiles $Q_{Z_1}(p)$ and $Q_{Z_2}(p)$ of $Z_1$ and $Z_2$, respectively, are such that
	\begin{align*}
	\log\biggl[\left(-{ \log(1-Q_{W_1}(p))\over \eta_1}\right)^{1/\sigma_1}\biggl]
	=
	Q_{Z_1}(p) 
	\end{align*}
	{and}
	\begin{align*}
	\log\biggl[\left(-{\log(1-Q_{W_2}(p))\over \eta_2}\right)^{1/\sigma_2}\biggl]
	=
	Q_{{\rho} Z_1+\sqrt{1-\rho^2} Z_2}(p)
	=
	Q_{Z_2}(p),
	\end{align*}
	where in the last equality we used that $\rho Z_1+\sqrt{1-\rho^2}Z_2$ and $Z_2$ have the same distribution (see Lemma \ref{Prop-dual-1}).
	Hence, the $p$-quantiles $Q_{W_1}(p)$ and $Q_{W_2}(p)$ are given by
	\begin{align*}
	Q_{W_1}(p)&=1-\exp\big[-\eta_1 \exp(\sigma_1 Q_{Z_1}(p))\big],
	\\[0,2cm]
	Q_{W_2}(p)&=
	1-\exp\big[-\eta_2 
	\exp\big(\sigma_2 Q_{Z_2}(p)\big)\big],
	\end{align*}
	respectively.

	\subsection{Conditional distributions}\label{cond-dist}
	
	%

Before enunciating and proving the main result (Theorem \ref{theo-pdf-cond}) of this subsection, we prove the following technical lemma which will be indispensable in the sequel.
	
	\begin{lemma}\label{conditional PDF}
		If $\boldsymbol{W}=(W_1,W_2)^\top\sim {\rm BULS}(\boldsymbol{\theta},g_c)$, then
		the PDF of $W_2 \,\vert\, (W_1=w_1)$ is given by
		\begin{align}\label{cond-pdf}
		f_{W_2}(w_2\,\vert\, W_1=w_1)
		=
		{1\over (1-w_2)t_2\sigma_2 \sqrt{1-\rho^2}}\,
		{
			f_{Z_2}\biggl(
			{1\over\sqrt{1-\rho^2}}\, (\widetilde{w}_2
			-
			\rho \widetilde{w}_1) \, \bigg\vert\, Z_1=\widetilde{w}_1 
			\biggr)
		},
		\end{align}
		where $\widetilde{w}_i$, $i=1,2$, and $t_2$ are as defined in \eqref{PDF-BULS}, and  ${Z_1}$ and ${Z_2}$  are as given in Proposition \ref{Stochastic Representation}.
	\end{lemma}
	\begin{proof}
		If $W_1=w_1$, then
		$
		Z_1= \log\big[\left(-{ \log(1-w_1)/ \eta_1}\right)^{1/\sigma_1}\big]
		=
		\widetilde{w}_1
		$.
		So, the conditional distribution of $W_2$, given $W_1=w_1$, is the same as the distribution of
		\begin{align*}
		1-\exp\left[-\eta_2 
		\exp\left(\sigma_2 {\rho} \widetilde{w}_1+\sigma_2\sqrt{1-\rho^2} Z_2\right)\right]\, \bigg\vert\, W_1=w_1.
		\end{align*}
		Consequently, 
		\begin{align*}
		F_{W_2}(w_2\, \vert\, W_1=w_1)
		&=
		\mathbb{P}\left(
		1-\exp\left[-\eta_2 
		\exp\left(\sigma_2 {\rho} \widetilde{w}_1+\sigma_2\sqrt{1-\rho^2} Z_2\right)\right]\leqslant w_2
		\, \bigg\vert\, W_1=w_1
		\right)
		\\[0,2cm]
		&=
		\mathbb{P}\biggl(
		Z_2\leqslant {1\over\sqrt{1-\rho^2}}\, (\widetilde{w}_2-{\rho} \widetilde{w}_1)\, \bigg\vert\, Z_1= \widetilde{w}_1\biggr).
		\end{align*}
		Then, by differentiating $F_{W_2}(w_2\, \vert\, W_1=w_1)$ with respect to $w_2$, \eqref{cond-pdf} is readily obtained.
	\end{proof}

The following result 
provides a simple formula for determining the conditional distribution of $W_1$, given $W_2\in B$, whenever the marginal and conditional distributions of $\boldsymbol{W}=(W_1,W_2)^\top\sim {\rm BULS}(\boldsymbol{\theta},g_c)$ are known.
	This result is essential for studying Heckman-type selection models \citep{Heckman1979} when the selection variables have unitary support.
	\begin{theorem}\label{theo-pdf-cond}
		For a Borelian subset $B$ of $(0,1)$, let us define the following Borelian set:
		\begin{align}\label{B-child}
		{B}_r
		=
		{1\over\sqrt{1-r^2}}\,
		\log\biggl[\left(-{\log(1-B)\over \eta_2}\right)^{1/\sigma_2}\biggr]
		-
		{r\over\sqrt{1-r^2}}\, \widetilde{w}_1, 
		\quad -1<r<1, 
		\end{align}
		where $\widetilde{w}_1$ is as in \eqref{PDF-BULS}.
		If $\boldsymbol{W}\sim {\rm BULS}(\boldsymbol{\theta},g_c)$, then the PDF of $W_1\,\vert\, (W_2\in B)$ is given by
		\begin{align*}
		f_{W_1}(w_1\, \vert\, W_2\in B)
		=
		{1\over (1-w_1)t_1 \sigma_1}\,f_{Z_1}(\widetilde{w}_1)\, 
		{
		\mathbb{P}(Z_2\in {B}_\rho\, \vert\, Z_1=\widetilde{w}_1 )
			\over 
			\mathbb{P}(Z_2\in B_0)
		},
		\end{align*}
		in which 
		$t_1$ is as in \eqref{PDF-BULS}, 
		${B}_r$ is as in \eqref{B-child},  and ${Z_1}$ and ${Z_2}$ are as given in Proposition \ref{Stochastic Representation}.
	\end{theorem}
	\begin{proof}
		Let $B$ be a Borelian subset of $(0,1)$. Note that
		\begin{align*}
		f_{W_1}(w_1 \, \vert\, W_2\in B)
		=
		f_{W_1}(w_1)\, {\int_B f_{W_2}(w_2\, \vert\, W_1=w_1)\, {\rm d}w_2\over \mathbb{P}(W_2\in B)}.
		\end{align*}
		As $f_{W_1}(w_1)=f_{Z_1}(\widetilde{w}_1)/[(1-w_1)t_1\sigma_1]$ and $\mathbb{P}(W_2\in B)=\mathbb{P}\big(\rho Z_1+\sqrt{1-\rho^2}Z_2\in {B}_0\big)$, where ${B}_0$ is as given in \eqref{B-child} with $r=0$, the term on the right-hand side of the above identity is
		\begin{align*}
		=
		{1\over (1-w_1)t_1\sigma_1 }\,f_{Z_1}(\widetilde{w}_1)\, {\int_B f_{W_2}(w_2\, \vert\, W_1=w_1)\, {\rm d}w_2\over \mathbb{P}\big(\rho Z_1+\sqrt{1-\rho^2}Z_2\in {B}_0\big)}.
		\end{align*}
		By using the formula for $f_{W_2}(w_2\vert W_1=w_1)$ provided in Lemma \ref{conditional PDF}, the above expression is
		\begin{align*}
		=
		{1\over (1-w_1)t_1 \sigma_1\sigma_2 \sqrt{1-\rho^2}}\,f_{Z_1}(\widetilde{w}_1)\, 
		{\int_B 
			{1\over (1-w_2)t_2}\,	
			f_{Z_2}\Big(
			{1\over\sqrt{1-\rho^2}}\,\widetilde{w}_2
			-
			{\rho\over\sqrt{1-\rho^2}}\, \widetilde{w}_1 \, \Big\vert\, Z_1=\widetilde{w}_1 
			\Big)\, 
			{\rm d}w_2 \over \mathbb{P}\big(\rho Z_1+\sqrt{1-\rho^2}Z_2\in {B}_0\big)},
		\end{align*}
		where $\widetilde{w}_i$ and $t_i$, $i=1,2$, are as in \eqref{PDF-BULS}.
		Finally, by applying the change of variable $z=(\widetilde{w}_2
		-
		{\rho}\, \widetilde{w}_1)/\sqrt{1-\rho^2}$, the above expression is
		\begin{align*}
		=
		{1 \over (1-w_1)t_1 \sigma_1}\,f_{Z_1}(\widetilde{w}_1)\, 
		{\int_{{B}_\rho}
			f_{Z_2}(
			z \, \vert\, Z_1=\widetilde{w}_1 )\, 
			{\rm d}z
			\over \mathbb{P}\big(\rho Z_1+\sqrt{1-\rho^2}Z_2\in {B}_0\big)}.
		\end{align*}
		
		We have thus proved that
		\begin{align*}
				f_{W_1}(w_1 \, \vert\, W_2\in B)
				=
					{1 \over (1-w_1)t_1 \sigma_1}\,f_{Z_1}(\widetilde{w}_1)\, 
				{\int_{{B}_\rho}
					f_{Z_2}(
					z \, \vert\, Z_1=\widetilde{w}_1 )\, 
					{\rm d}z
					\over \mathbb{P}\big(\rho Z_1+\sqrt{1-\rho^2}Z_2\in {B}_0\big)}.
		\end{align*}
		Finally, by combining the above identity with Lemma \ref{Prop-dual-1}, the required result follows.
	\end{proof}

Using Theorem \ref{theo-pdf-cond}, for each generator ($g_c$) in Table \ref{table:1}, we present closed formulas for the conditional densities of $W_1\,\vert\, (W_2\in B)$ corresponding to
bivariate unit-log-normal (Corollary \ref{Gaussian generator}), 
bivariate unit-log-Student-$t$ (Corollary \ref{Student-t-generator}),
bivariate unit-log-hyperbolic (Corollary \ref{Hyperbolic generator}),  
bivariate unit-log-Laplace (Corollary \ref{Laplace generator}) and  
bivariate unit-log-slash (Corollary \ref{Slash generator}) distributions.
	
	\begin{corollary}[Gaussian generator] \label{Gaussian generator}
		Let $\boldsymbol{W}=(W_1,W_2)^\top\sim {\rm BULS}(\boldsymbol{\theta},g_c)$ and $g_c(x)=\exp(-x/2)$ be the generator of the bivariate unit-log-normal distribution. Then, for each Borelian subset $B$ of $(0,1)$, the PDF of $W_1\,\vert\, (W_2\in B)$ is given by (for $0<w_1<1$)
		\begin{align*}
		f_{W_1}(w_1\,\vert\, W_2\in B)
		=
		{1\over (1-w_1)t_1 \sigma_1}\,
		\phi(\widetilde{w}_1)\, 
		\dfrac{
			\Phi(B_\rho)
		}{
			\Phi(B_0)
		},
		\end{align*}
		where $\Phi(C)=\int_C \phi(x){\rm d}x$ and $\phi(x)$ is the standard normal PDF. Further, $\widetilde{w}_1$ and $t_1$ are as in \eqref{PDF-BULS}, and $B_r$ is as in \eqref{B-child}.
	\end{corollary}
	\begin{proof}
		It is well-known that
		the bivariate log-normal distribution has a stochastic representation as in \eqref{rep-stoch-biv-gaussian}, where 
		$Z_1\sim N(0,1)$ and $Z_2\sim N(0,1)$, and $Z_2\, \vert\, (Z_1=x)\sim N(0,1)$ \citep{Abdous2005}.
Hence,
$
\mathbb{P}(Z_2\in B_0)
=
\Phi(B_0)
$
and 
$
\mathbb{P}(Z_2\in {B}_\rho\, \vert\, Z_1=\widetilde{w}_1 )
=
\Phi({B}_\rho).
$
		Then, by applying Theorem \ref{theo-pdf-cond}, the required result follows.
	\end{proof}

	\begin{corollary}[Student-$t$ generator]\label{Student-t-generator}
		Let $\boldsymbol{W}=(W_1,W_2)^\top\sim {\rm BULS}(\boldsymbol{\theta},g_c)$ and $g_c(x)=(1+(x/\nu))^{-(\nu+2)/2}$, $\nu>0$, be the generator of the bivariate unit-log-Student-$t$ distribution with $\nu$ degrees of freedom. Then, for each Borelian subset $B$ of $(0,1)$, the PDF of $W_1\,\vert\, (W_2\in B)$ is given by (for $0<w_1<1$)
		\begin{align*}
		f_{W_1}(w_1\,\vert\, W_2\in B)
		=
		{1 \over (1-w_1) t_1 \sigma_1}\,
		f_{\nu}(\widetilde{w}_1)\, 
		\dfrac{
			F_{\nu+1}\Big(\sqrt{\nu+1\over \nu+\widetilde{w}_1^2}\, 
			B_\rho 
			\Big)		
		}{
			F_\nu(B_0)
		},
		\end{align*}
		where  $F_\nu(C)=\int_C f_\nu(x){\rm d}x$ and $f_\nu(x)$ is the standard Student-$t$ PDF with $\nu$ degrees of freedom.
	\end{corollary}
	\begin{proof}
	It is well-known that
	the bivariate log-Student-$t$ distribution has a stochastic representation as in \eqref{rep-stoch-biv-gaussian}, where 
	$Z_1\sim t_\nu$ and $Z_2\sim t_\nu$ 
	(Student-$t$ with $\nu$ degrees of freedom), and \cite[see Corollary 3.7 of][]{Vila2022}
	$$
	Z_2\, \vert\, (Z_1=x) \sim \sqrt{\nu+x^2\over \nu+1}\, t_{\nu+1}.
	$$ 
	Hence,	
	%
$\mathbb{P}(Z_2\in B_0)
=
F_{\nu}(B_0)$ and
		\begin{align*}
		\mathbb{P}(Z_2\in {B}_\rho\, \vert\, Z_1=\widetilde{w}_1 )
		=
		F_{\nu+1}\left(\sqrt{\nu+1\over \nu+\widetilde{w}_1^2}\, {B}_\rho \right).
		\end{align*}
		By applying Theorem \ref{theo-pdf-cond}, the required result follows.
	\end{proof}

\begin{corollary}[Hyperbolic generator]\label{Hyperbolic generator}
	Let $\boldsymbol{W}=(W_1,W_2)^\top\sim {\rm BULS}(\boldsymbol{\theta},g_c)$ and $g_c(x)=\exp(-\nu\sqrt{1+x}\,)$ be the generator of the bivariate unit-log-hyperbolic distribution. Then, for each Borelian subset $B$ of $(0,1)$, the PDF of $W_1\,\vert\, (W_2\in B)$ is given by (for $0<w_1<1$)
	\begin{align*}
f_{W_1}(w_1\,\vert\, W_2\in B)
=
{1 \over (1-w_1) t_1 \sigma_1}\,
f_{\rm GH}(\widetilde{w}_1; 3/2,\nu,1)\, 
\dfrac{
	F_{\rm GH}\left({B}_\rho; 1,\nu,\sqrt{1+\widetilde{w}_1^2}\right)		
}{
	F_{\rm GH}(B_0;3/2,\nu,1)
},
\end{align*}
where  $F_{\rm GH}(C;\lambda,\alpha,\delta)=\int_C f_{\rm GH}(x;\lambda,\alpha,\delta){\rm d}x$ and $f_{\rm GH}(x;\lambda,\alpha,\delta)$ is the generalized hyperbolic (GH) PDF (see Definition \ref{def-GH} in the Appendix).
\end{corollary}
\begin{proof}
	It is well-known that
the bivariate log-hyperbolic distribution has a stochastic representation as in \eqref{rep-stoch-biv-gaussian}, where 
$Z_1
\sim {\rm GH}(3/2,\nu,1)$ 
and 
$Z_2
\sim {\rm GH}(3/2,\nu,1)$ 
\cite[Subsection 2.1, p. 3]{dy:18}. 
%
Moreover, the distribution of $Z_{2}$, given $Z_{1}=x$, is ${\rm GH}(1/2,\sqrt{2},\vert x\vert)$ (Proposition \ref{cond-marg-Hyp}). Then
$\mathbb{P}(Z_2\in B_0)
=
F_{\rm GH}(B_0;3/2,\nu,1)$ and
$
\mathbb{P}(Z_2\in {B}_\rho\, \vert\, Z_1=\widetilde{w}_1 )
=
	F_{\rm GH}\big({B}_\rho; 1,\nu,\sqrt{1+\widetilde{w}_1^2}\,\big).
$
By applying Theorem \ref{theo-pdf-cond}, the required result follows.
\end{proof}

	\begin{corollary}[Laplace generator]\label{Laplace generator}
	Let $\boldsymbol{W}=(W_1,W_2)^\top\sim {\rm BULS}(\boldsymbol{\theta},g_c)$ and $g_c(x)=K_0(\sqrt{2x})$ be the generator of the bivariate unit-log-Laplace distribution. Then, for each Borelian subset $B$ of $(0,1)$, the PDF of $W_1\,\vert\, (W_2\in B)$ is given by (for $0<w_1<1$)
	\begin{align*}
	f_{W_1}(w_1\,\vert\, W_2\in B)
	=
	{1 \over (1-w_1) t_1 \sigma_1}\,
	f_{\rm L}(\widetilde{w}_1)\, 
	\dfrac{
F_{\rm GH}({B}_\rho; {1\over 2},\sqrt{2},\vert \widetilde{w}_1 \vert)		
	}{
		F_{\rm L}(B_0)
	},
	\end{align*}
	where  $F_{\rm L}(C)=\int_C f_{\rm L}(x){\rm d}x$ and $f_{\rm L}(x)=\exp(-\sqrt{2}\, \vert x\vert)/\sqrt{2}$ is the Laplace PDF with scale parameter $1/\sqrt{2}$, and $F_{\rm GH}$ is as defined in Corollary \ref{Hyperbolic generator}.
\end{corollary}
\begin{proof}
	It is well-known that
the bivariate log-Laplace distribution has a stochastic representation as in \eqref{rep-stoch-biv-gaussian}, where 
$Z_1
\sim {\rm Laplace}(0,1/\sqrt{2})
$ 
and 
$Z_2
\sim {\rm Laplace}(0,1/\sqrt{2})$ \cite[Subsection 5.1.4, p. 234]{Kotz2001}. 
%
Further, the distribution of $Z_{2}$, given $Z_{1}=x$, is ${\rm GH}(1/2,\sqrt{2},\vert x\vert)$ (Proposition \ref{cond-marg-Laplace}). Hence,
$
\mathbb{P}(Z_2\in B_0)
=
F_{\rm L}(B_0)
$
and
$
\mathbb{P}(Z_2\in {B}_\rho\, \vert\, Z_1=\widetilde{w}_1 )
=
F_{\rm GH}({B}_\rho; {1/ 2},\sqrt{2},\vert \widetilde{w}_1 \vert).
$
By applying Theorem \ref{theo-pdf-cond}, the required result follows.
\end{proof}

	\begin{corollary}[Slash generator]\label{Slash generator}
	Let $\boldsymbol{W}=(W_1,W_2)^\top\sim {\rm BULS}(\boldsymbol{\theta},g_c)$ and $g_c(x)=x^{-{(q+2)/ 2}} \gamma({(q+2)/ 2},{x/ 2})$, be the generator of the bivariate unit-log-slash distribution. Then, for each Borelian subset $B$ of $(0,1)$, the PDF of $W_1\,\vert\, (W_2\in B)$ is given by (for $0<w_1<1$)
	\begin{align*}
	f_{W_1}(w_1\,\vert\, W_2\in B)
	=
	{1 \over (1-w_1) t_1 \sigma_1}\,
	f_{\rm SL}(\widetilde{w}_1;q)\, 
	\dfrac{
	F_{\rm ESL}\left({B}_\rho; \widetilde{w}_1 ,q+1\right)		
	}{
		F_{\rm SL}(B_0;q)
	},
	\end{align*}
	where  $F_{\rm SL}(C;q)=\int_C f_{\rm SL}(x;q){\rm d}x$ and $f_{\rm SL}(x;q)=q \int_{0}^1 t^q \phi(t x) {\rm d}t$ is the classical slash PDF, and $F_{\rm ESL}(C;a,q) =\int_C f_{\rm ESL}(x;a,q){\rm d}x$, where $f_{\rm ESL}(x;a,q)$ 
	is the 
	generalized hyperbolic (ESL) PDF (see Definition \ref{def-GH-1} in the Appendix).
\end{corollary}
\begin{proof}
	It is well-known that 
	the bivariate log-slash distribution has a stochastic representation as in \eqref{rep-stoch-biv-gaussian}, where 
	$Z_1
	\sim {\rm SL}(q)$ 
	and 
	$Z_2
	\sim {\rm SL}(q)$ \cite[Section 2, p. 211]{wg:06}. 
%
	Moreover, the distribution of $Z_{2}$, given $Z_{1}=x$, is ${\rm ESL}(x,q+1)$ (Proposition \ref{cond-marg-slash}). Hence,
	$
	\mathbb{P}(Z_2\in B_0)
	=
	F_{\rm SL}(B_0;q)
	$
	and
	$
	\mathbb{P}(Z_2\in {B}_\rho\, \vert\, Z_1=\widetilde{w}_1 )
	=
	F_{\rm ESL}\left({B}_\rho; \widetilde{w}_1 ,q+1\right).
	$
	By applying Theorem \ref{theo-pdf-cond}, the required result follows.
\end{proof}

	Table~\ref{table:2} below presents some examples of conditional PDFs corresponding to all the bivariate unit-log-symmetric distributions presented in Table \ref{table:1}.
	\begin{table}[H]
	\caption{Conditional densities of $W_1\vert\, (W_2\in B)$ and density generators $(g_c)$ for some BULS distributions.}
	\vspace*{0.15cm}
	\centering 
	\begin{tabular}{llll} 
		\hline
		Distribution 
		& $g_c$ & $f_{W_1}(w_1\,\vert\, W_2\in B)$ 
		\\ [0.5ex] 
		\noalign{\hrule}
		Bivariate unit-log-normal
		& $\exp(-x/2)$ & 
		$		{1\over (1-w_1)t_1 \sigma_1}\,
		\phi(\widetilde{w}_1)\, 
		\frac{
			\Phi(B_\rho)
		}{
			\Phi(B_0)
		}$
		\\ [1ex] 
		Bivariate unit-log-Student-$t$
		& $(1+{x\over\nu})^{-(\nu+2)/ 2}$  & 
		$	{1 \over (1-w_1) t_1 \sigma_1}\,
		f_{\nu}(\widetilde{w}_1)\, 
		\frac{
			F_{\nu+1}\Big(\sqrt{\nu+1\over \nu+\widetilde{w}_1^2}\, 
			B_\rho 
			\Big)		
		}{
			F_\nu(B_0)
		}$ 
		\\ [1ex]
		Bivariate unit-log-hyperbolic
		& $\exp(-\nu\sqrt{1+x})$ & 
		${1 \over (1-w_1) t_1 \sigma_1}\,
		f_{\rm GH}(\widetilde{w}_1; 3/2,\nu,1)\, 
		\frac{
			F_{\rm GH}\big({B}_\rho; 1,\nu,\sqrt{1+\widetilde{w}_1^2}\,\big)		
		}{
			F_{\rm GH}(B_0;3/2,\nu,1)
		}$ 
		\\ [1ex]   
		Bivariate unit-log-Laplace
		& $K_0(\sqrt{2x})$ & 	${1 \over (1-w_1) t_1 \sigma_1}\,
		f_{\rm L}(\widetilde{w}_1)\, 
		\frac{
			F_{\rm GH}({B}_\rho; {1\over 2},\sqrt{2},\vert \widetilde{w}_1 \vert)		
		}{
			F_{\rm L}(B_0)
		}$
		\\ [1ex]   
		Bivariate unit-log-slash
		& $ x^{-{q+2\over 2}} \gamma({q+2\over 2},{x\over 2})$ & $	{1 \over (1-w_1) t_1 \sigma_1}\,
		f_{\rm SL}(\widetilde{w}_1;q)\, 
		\frac{
			F_{\rm ESL}({B}_\rho; \widetilde{w}_1 ,q+1)		
		}{
			F_{\rm SL}(B_0;q)
		}$
		\\ [1ex]   	
		\hline	
	\end{tabular}
	\label{table:2} 
\end{table}

	%

	\subsection{Squared Mahalanobis Distance}\label{maha_sec}
	
	The squared Mahalanobis distance of a random vector $\boldsymbol{W}=(W_1, W_2)^\top$ and the vector $\log(\boldsymbol{\eta})=(\log(\eta_1),\log(\eta_2))^{\top}$ of a BULS distribution is defined as
	\begin{eqnarray*}
		d^2(\boldsymbol{W}, \log(\boldsymbol{\eta}))
		=
		{\widetilde{W}_1^2-2\rho\widetilde{W}_1\widetilde{W}_2+\widetilde{W}_2^2
			\over 
			1-\rho^2},
	\end{eqnarray*}
where  $\widetilde{W}_i
		=
		\log[({T_i/ \eta_i})^{1/\sigma_i}], \,
		T_i=-\log(1-W_i),
		$ and $
		\eta_i=\exp(\mu_i), 
		\ i=1,2$.
Then, analogous to Propositions 3.8 and 3.9 of \cite{Vila2022}, we have
the following  formulas for the CDF and PDF of the random variable $d^2(\boldsymbol{W}, \log(\boldsymbol{\eta}))$:
		\begin{align*}
		F_{d^2(\boldsymbol{W}, \log(\boldsymbol{\eta}))}(x)
		&=
		\displaystyle 
		{4\over Z_{g_c}}\, 
		\int_{0}^{\sqrt{x}}
		\left[\int_{0}^{\sqrt{x-z_1^2}} g_c(z_1^2+z_2^2) \, {\rm d}z_2\right] {\rm d}z_1,  \quad x> 0,
		\\[0,2cm]
		f_{d^2(\boldsymbol{W}, \log(\boldsymbol{\eta}))}(x)
		&={\pi\over Z_{g_c}}\, g_c(x), \quad x>0,
		\end{align*}
		where $Z_{g_c}$ is as in \eqref{partition function}.
	
	For example, upon taking $g_c(x)=\exp(-x/2)$ and $Z_{g_c}=2\pi$ (see Table \ref{table:1}),  we get
	$d^2(\boldsymbol{W}, \log(\boldsymbol{\eta}))\sim \chi^2_2$ (chi-square with $2$ degrees of freedom). 
		Next, upon taking $g_c(x)=(1+(x/\nu))^{-(\nu+2)/2}$ and $Z_{g_c}={{\Gamma({\nu/ 2})}\nu\pi/{\Gamma({(\nu+2)/ 2})}}$ (see Table \ref{table:1}), we have
			$d^2(\boldsymbol{W}, \log(\boldsymbol{\eta})) \sim 2 F_{2,\nu}$, where $F_{2,\nu}$ denotes the F-distribution with $2$ and $\nu$ degrees of freedom.

	\subsection{Independence}\label{ind}
	
	\begin{proposition}
		Let $\boldsymbol{W}=(W_1,W_2)^\top\sim {\rm BULS}(\boldsymbol{\theta},g_c)$. If $\rho=0$ and the density generator $g_c$ in \eqref{PDF-BULS} is such that
		\begin{align}\label{kernel-dec}
		g_c\big(x^2+y^2\big)
		=g_{c_1}\big(x^2\big) 
		g_{c_2}\big(y^2\big),
		\quad \forall (x,y)\in\mathbb{R}^2,
		\end{align}
		for some density generators $g_{c_1}$ and $g_{c_2}$, then $W_1$ and $W_2$ are independent.
	\end{proposition}
	\begin{proof}
		The proof follows the same steps as the proof of Proposition 3.11 of \cite{Vila2022}. For the sake of completeness, however, we present it here.

		Let $\rho=0$. From \eqref{kernel-dec}, the joint density \eqref{PDF-BULS} of $(W_1,W_2)$ is such that
		\begin{align}\label{pdf-product}
		f_{W_1,W_2}(w_1,w_2;\boldsymbol{\theta})
		=
		{Z_{g_{c_1}}Z_{g_{c_2}}\over Z_{g_c}}\,
		f_1(w_1;\mu_1,\sigma_1)
		f_2(w_2;\mu_2,\sigma_2), \quad\forall(w_1, w_2)\in(0,1)\times (0,1),
		\end{align}
		where
		$
		f_i(w_i;\mu_i,\sigma_i)=
		g_{c_i}\bigl(\widetilde{w_i}^2\big)/[(1-w_i)t_i\sigma_iZ_{g_{c_i}}], \ 0<w_i<1$, $Z_{g_{c_i}}
		=
		\int_{-\infty}^{\infty}
		g_{c_i}\big({z_i}^2\big)\, {\rm d}z_i, \ i=1,2,
		$
		and $\widetilde{w_i}$ and $t_i$ are as in \eqref{PDF-BULS}.
		Integrating \eqref{pdf-product} in terms of $w_1$ and $w_2$, we obtain
		\begin{align*}
		{Z_{g_{c_1}}Z_{g_{c_2}}\over Z_{g_c}}=1,
		\end{align*}
		and consequently, $Z_{g_c}=Z_{g_{c_1}}Z_{g_{c_2}}$. Therefore,
		\begin{align*}
		f_{W_1,W_2}(w_1,w_2;\boldsymbol{\theta})
		=
		f_1(w_1;\mu_1,\sigma_1)
		f_2(w_2;\mu_2,\sigma_2), \quad\forall(w_1, w_2)\in(0,1)\times (0,1).
		\end{align*}
		Moreover, it is easy to verify that $f_1$ and $f_2$ are PDFs corresponding to univariate symmetric random variables \citep{Vanegas2016}. Then, $W_1$ and $W_2$ are statistically independent, and even more, $f_i=f_{W_i}$, for $i=1,2$  \cite[see Proposition 2.5 of][]{James2004}.
	\end{proof}
	
	\begin{remark}
		In Table \ref{table:1}, the density generator of the bivariate unit-log-normal is the unique one that  satisfies \eqref{kernel-dec}.
	\end{remark}

	\subsection{Moments}\label{moments}

	For $\boldsymbol{W}=(W_1,W_2)^\top\sim {\rm BULS}(\boldsymbol{\theta},g_c)$, $0<W_i<1$, it is clear that  $0\leqslant\mathbb{E}(W_i^r)\leqslant 1$, for any $r>0$ and $i=1,2$.
	Therefore, the positive moments of $W_i$ always exist.
	
	In general, for any $r\in\mathbb{R}$, the moments of $W_i$, $i=1,2$, admit the following representations:
	\begin{align*}
	\begin{array}{lllll}
	\mathbb{E}(W_1^r)
	&=
	\displaystyle
	\mathbb{E}
	\left\{
	1-\exp\big[-\eta_1 \exp(\sigma_1 Z_1)\big]
	\right\}^r, 
	\\[0,5cm]
	\mathbb{E}(W_2^r)
	&=
	\displaystyle
	\mathbb{E}
	\big(
	1-\exp\big\{-\eta_2 
	\exp\big(\sigma_2 \big[{{\rho} Z_1+\sqrt{1-\rho^2} Z_2}\big]\big)\big\}
	\big)^r
	=
	\displaystyle
	\mathbb{E}
	\left\{
	1-\exp\big[-\eta_2 \exp(\sigma_2 Z_2)\big]
	\right\}^r,
	\end{array}
	\end{align*}
	where in the last equality we used that $\rho Z_1+\sqrt{1-\rho^2}Z_2$ and $Z_2$ have the same distribution (see Lemma \ref{Prop-dual-1}).
	Here, $Z_1$ and $Z_2$ are as given in Proposition \ref{Stochastic Representation}.

	\section{Maximum likelihood estimation} \label{Sec:5}
	\noindent
	Let $\{(W_{1i},W_{2i})^\top:i=1,\ldots,n\}$ be a bivariate random sample of size $n$ from the ${\rm BULS}(\boldsymbol{\theta},g_c)$ distribution with PDF as in \eqref{PDF-BULS}, and let $(w_{1i},w_{2i})^\top$ be the corresponding observations of $(W_{1i},W_{2i})^\top$. Then, the log-likelihood function for $\boldsymbol{\theta}=(\eta_1,\eta_2,\sigma_1,\sigma_2,\rho)^{\top}$,
	without the additive constant, is given by
	%
	%
	\begin{align*}
	\ell(\boldsymbol{\theta})
	=
	-n\sum_{i=1}^{2}\log(\sigma_i)
	-
	{n\over 2}\,\log\big({1-\rho^2}\big)
	+
	\sum_{i=1}^{n}
	\log	
	g_c\Biggl(
	{\widetilde{w}_{1i}^2-2\rho\widetilde{w}_{1i}\widetilde{w}_{2i}+\widetilde{w}_{2i}^2
		\over 
		1-\rho^2}
	\Biggr), \quad
	0<w_{1i},w_{2i}<1,
\end{align*}
where $
	\widetilde{w}_{ki}=\log[({t_{ki}/ \eta_k})^{1/\sigma_k}],  
	\ t_{ki}=-\log(1-w_{ki})>0
	$ and $
	\eta_k=\exp(\mu_k), \ k=1,2; \ i=1,\ldots,n.$

	In the case when a supremum $\widehat{\boldsymbol{\theta}}=(\widehat{\eta_1},\widehat{\eta_2},\widehat{\sigma_1},\widehat{\sigma_2},\widehat{\rho})^{\top}$ exists, it must satisfy the following likelihood equations:
	\begin{align}\label{likelihood equation}
	{\partial \ell({\boldsymbol{\theta}})\over\partial\eta_1}
	\bigg\vert_{{\boldsymbol{\theta}}=\widehat{\boldsymbol{\theta}}}
	=0,
	\quad 
	{\partial \ell(\boldsymbol{\theta})\over\partial\eta_2}=0,
	\quad 
	{\partial\ell(\boldsymbol{\theta})\over\partial\sigma_1}
	\bigg\vert_{{\boldsymbol{\theta}}=\widehat{\boldsymbol{\theta}}}=0,
	\quad 
	{\partial\ell(\boldsymbol{\theta})\over\partial\sigma_2}
	\bigg\vert_{{\boldsymbol{\theta}}=\widehat{\boldsymbol{\theta}}}=0,
	\quad 
	{\partial\ell(\boldsymbol{\theta})\over\partial\rho}
	\bigg\vert_{{\boldsymbol{\theta}}=\widehat{\boldsymbol{\theta}}}=0,
	\end{align}
	with
	\begin{align}
	&{\partial \ell(\boldsymbol{\theta})\over\partial\eta_1}
	=
	\frac{2}{\sigma_1\eta_1(1-\rho^2)}
	\sum_{i=1}^{n}
	\big(\rho \widetilde{w}_{2i} -\widetilde{w}_{1i}\big)
	G(\widetilde{w}_{1i},\widetilde{w}_{2i})
	, \nonumber
	\\[0,1cm]
	&{\partial \ell(\boldsymbol{\theta})\over\partial\eta_2}
	=
	\frac{2}{\sigma_2\eta_2(1-\rho^2)}
	\sum_{i=1}^{n}
	\big(\rho \widetilde{w}_{1i} -\widetilde{w}_{2i}\big)
	G(\widetilde{w}_{1i},\widetilde{w}_{2i})
	, \nonumber
	\\[0,1cm]
	&{\partial\ell(\boldsymbol{\theta})\over\partial\sigma_1}
	=
	-\frac{n}{\sigma_1}
	+
	{2\over \sigma_1(1-\rho^2)}
	\sum_{i=1}^{n}
	\widetilde{w}_{1i}
	\big(\rho\widetilde{w}_{2i}-\widetilde{w}_{1i}\big)
	G(\widetilde{w}_{1i},\widetilde{w}_{2i})
	, \nonumber
	\\[0,1cm]
	&{\partial\ell(\boldsymbol{\theta})\over\partial\sigma_2}
	=
	-\frac{n}{\sigma_2}
	+
	{2\over \sigma_2(1-\rho^2)}
	\sum_{i=1}^{n}
	\widetilde{w}_{2i}
	\big(\rho\widetilde{w}_{1i}-\widetilde{w}_{2i}\big)
	G(\widetilde{w}_{1i},\widetilde{w}_{2i})
	, \nonumber
	\\[0,1cm]
	&{\partial \ell(\boldsymbol{\theta})\over\partial\rho}
	=
	{n\rho\over 1-\rho^2}
	-
	{2\over (1-\rho^2)^2}
	\sum_{i=1}^{n}
	\big(\rho \widetilde{w}_{1i}-\widetilde{w}_{2i}\big)
	\big(\rho\widetilde{w}_{2i}-\widetilde{w}_{1i}\big)
	%
	G(\widetilde{w}_{1i},\widetilde{w}_{2i}), \label{rho-mle}
	\end{align}
	where we have used the notation
	\begin{align}\label{def-x-rho}
	G(\widetilde{w}_{1i},\widetilde{w}_{2i})
	=
	{g_c'(x_{\rho,i})\over g_c(x_{\rho,i})},
	\end{align}
	with $x_{\rho,i}=	{(\widetilde{w}_{1i}^2-2\rho\widetilde{w}_{1i}\widetilde{w}_{2i}+\widetilde{w}_{2i}^2)
		/
		(1-\rho^2)},
	\ i=1,\ldots,n.$
	
	Observe that the likelihood equations in \eqref{likelihood equation} can be written as
	\begin{align*}
	&\sum_{i=1}^{n}
	\widetilde{w}_{1i}\,
	G(\widetilde{w}_{1i},\widetilde{w}_{2i})
	\bigg\vert_{{\boldsymbol{\theta}}=\widehat{\boldsymbol{\theta}}}
	=0,
	\\[0,1cm]
	&
	\sum_{i=1}^{n}
	\big(\widetilde{w}_{1i}^2-\widetilde{w}_{2i}^2\big)\,
	G(\widetilde{w}_{1i},\widetilde{w}_{2i})
	\bigg\vert_{{\boldsymbol{\theta}}=\widehat{\boldsymbol{\theta}}}
	=0,
	\\[0,2cm]
	&
	\sum_{i=1}^{n}
	\widetilde{w}_{2i}
	\left[2\rho
	\widetilde{w}_{2i}
	-
	(1+\rho^2) \widetilde{w}_{1i}\right]
	G(\widetilde{w}_{1i},\widetilde{w}_{2i})
	\bigg\vert_{{\boldsymbol{\theta}}=\widehat{\boldsymbol{\theta}}}
	=-{n\widehat{\rho}(1-\widehat{\rho}^2)\over 2}\,.
	\end{align*}
	Any nontrivial root $\widehat{\boldsymbol{\theta}}$ of the above likelihood equations is an ML estimator in the loose sense. When the parameter value provides the absolute maximum of the log-likelihood function, it becomes the ML estimator in the strict sense.
	
	In the following proposition, we discuss the existence of the ML estimator $\widehat{\rho}$ when all other parameters are known.
	\begin{proposition}\label{prop-existence-MLE}
		Let $g_c$ be a density generator such that
		\begin{align}\label{condition-g}
		g'_c(x)=r(x) g_c(x),\quad -\infty<x<\infty,
		\end{align}
		for some real-valued function $r(x)$ with $\lim_{\rho\to \pm 1} r(x_{\rho,i})=c\in(-\infty,0)$, where $x_{\rho,i}$, $i=1,\ldots,n$, are as in \eqref{def-x-rho}.
		If the parameters $\eta_1,\eta_2,\sigma_1$ and $\sigma_2$ are all known, then \eqref{rho-mle} has at least one root in the interval $(-1, 1)$.
	\end{proposition}
	\begin{proof}
		The proof of this result follows by direct application of Intermediate value theorem. For more details, see  Proposition 5.1. of \cite{Vila2022}.
	\end{proof}
	%
	
	For the BULS model, no closed-form solution to the maximization problem is available, and an MLE can only be found by means of numerical optimization. Under mild regularity conditions \citep{Cox1974,Davison2008}, the asymptotic distribution of the ML estimator  $\widehat{\boldsymbol{\theta}}$ of $\boldsymbol{\theta}$ is as follows: $(\widehat{\boldsymbol{\theta}}-\boldsymbol{\theta})\stackrel{\mathscr D}{\longrightarrow} N(\boldsymbol{0},I^{-1}(\boldsymbol{\theta}))$,
	where  
	$\boldsymbol{0}$ is the zero mean vector and $I^{-1}(\boldsymbol{\theta})$ is the inverse expected 
	Fisher information matrix.
	The main use of the last convergence is to construct confidence regions and to perform
	hypothesis testing for $\boldsymbol{\theta}$ \citep{Davison2008}.

\section{Simulation study} \label{Sec:4}
\noindent

In this section, we carry out a Monte Carlo simulation study for evaluating the performance of the ML estimators of the parameters of BULS distributions.  For illustration purposes, we only present results for the bivariate unit-log-normal model. The simulation scenario considers is as follows: 1,000 Monte Carlo replications, sample size $n \in (25,100,500,700)$, vector of true parameters $(\eta_1,\eta_2,\sigma_1,\sigma_2)= (1,1,0.5,0.5)$, $\rho \in \{0,0.25,0.5,0.75,0.95\}$ (negative values of $\rho$ produce the same results and so are omitted). To study the performance of the ML estimators, we computed the bias, root mean square error (RMSE), and coverage probability (CP), defined by
\begin{eqnarray*}
 \widehat{\textrm{Bias}}(\widehat{\theta}) &=&  \frac{1}{N} \sum_{i = 1}^{N} \widehat{\theta}^{(i)} - \theta ,\\
\widehat{\mathrm{RMSE}}(\widehat{\theta}) &=& {\sqrt{\frac{1}{N} \sum_{i = 1}^{N} (\widehat{\theta}^{(i)} - \theta)^2}}, \\
\widehat{\mathrm{CP}}(\widehat{\theta}) &=& \frac{1}{N} \sum_{i = 1}^{N} \mathcal{I}(\theta \in [L^{(i)}_{\widehat{\theta}},U^{(i)}_{\widehat{\theta}}]),
\end{eqnarray*}
where $\theta$ and $\widehat{\theta}^{(i)}$ are the true parameter value and its $i$-th ML estimate, $N$ is the number of Monte Carlo replications,  $\mathcal{I}$ is an indicator function taking the value 1 if $\theta\in \left[L^{(i)}_{\widehat{\theta}},U^{(i)}_{\widehat{\theta}}\right]$, and 0 otherwise, where $L^{(i)}_{\widehat{\theta}}$ and $U^{(i)}_{\widehat{\theta}}$ are the $i$-th upper and lower limit estimates of the 95\% confidence interval. We expect that, as the sample size increases, the bias and RMSE would decrease, and the CP would approach the 95\% nominal level.

The obtained simulation results are presented in Figure \ref{fig_normal_mc}. We observe that the results obtained for the chosen bivariate unit-log-normal distribution are as expected in that as the sample size increases, the bias and RMSE both decrease and that the CP approaches the 95\% nominal level. Finally, in general, the results do not seem to depend on the parameter $\rho$.

\begin{figure}[H]
\vspace{-0.25cm}
\centering
{\includegraphics[height=3.5cm,width=3.5cm]{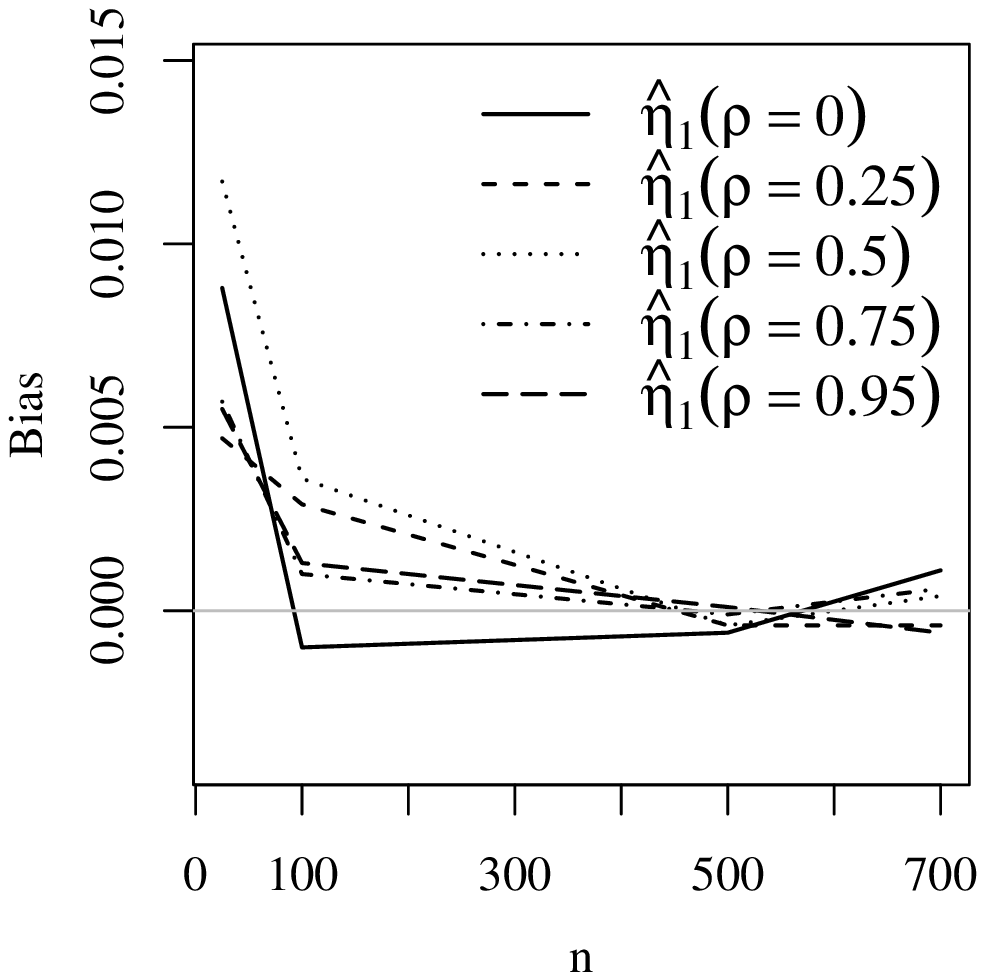}}\hspace{-0.25cm}
{\includegraphics[height=3.5cm,width=3.5cm]{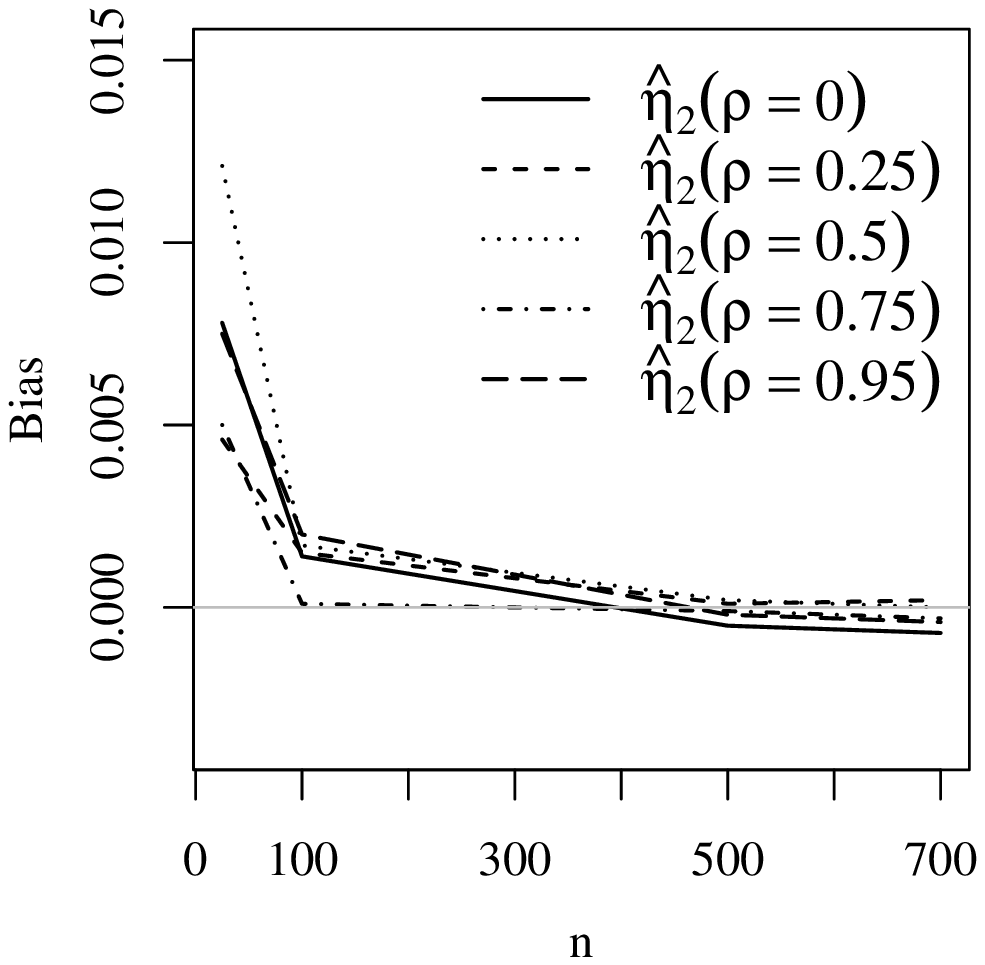}}\hspace{-0.25cm}
{\includegraphics[height=3.5cm,width=3.5cm]{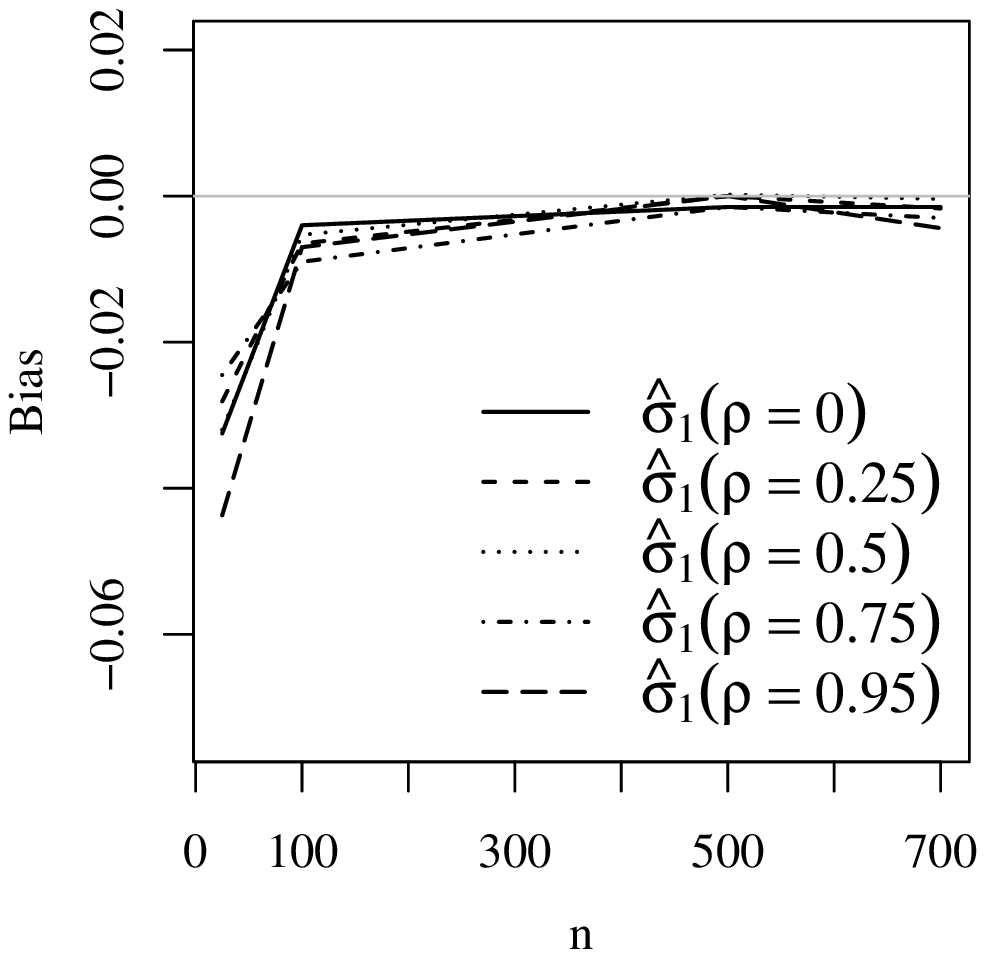}}\hspace{-0.25cm}
{\includegraphics[height=3.5cm,width=3.5cm]{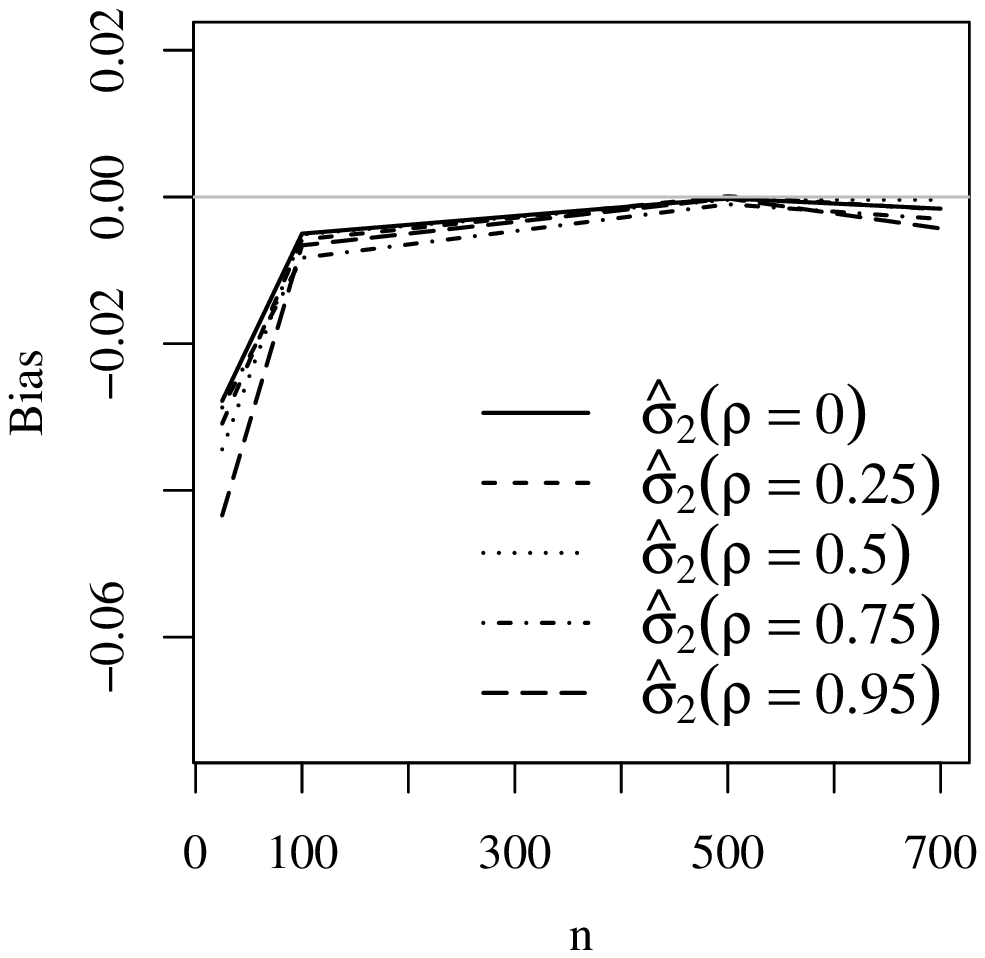}}\hspace{-0.25cm}
{\includegraphics[height=3.5cm,width=3.5cm]{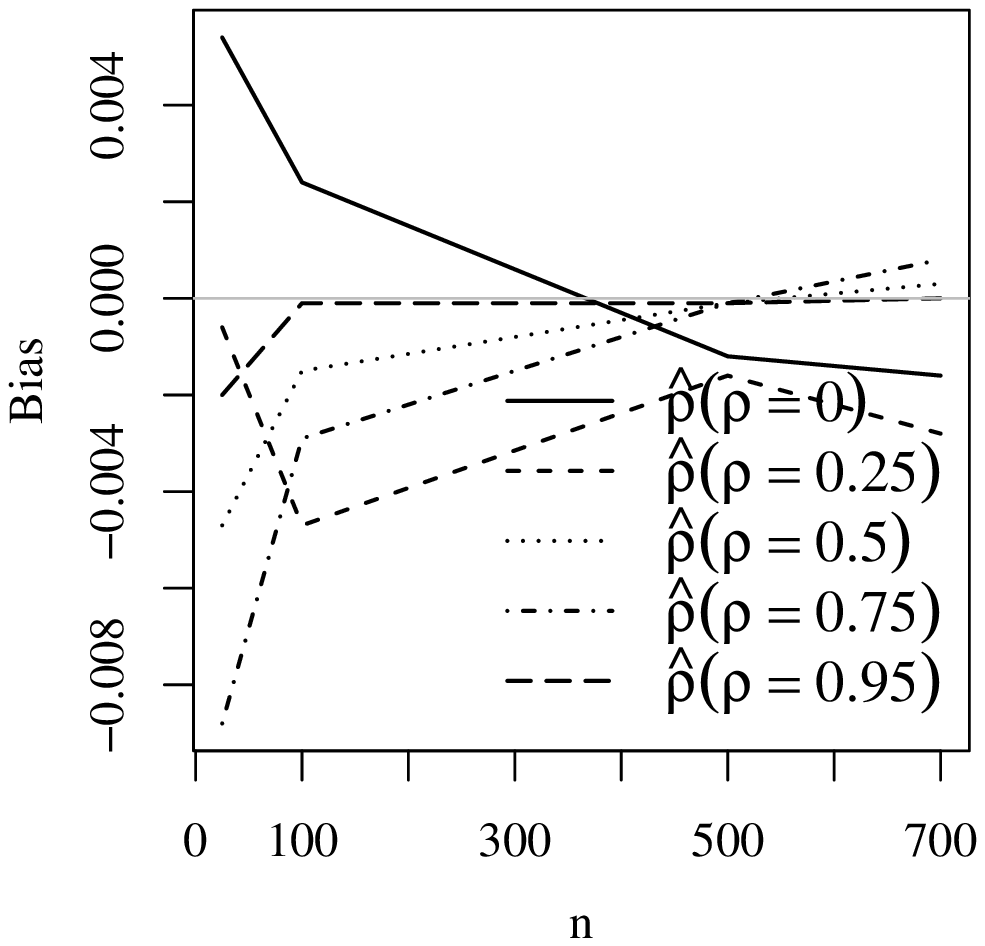}}
{\includegraphics[height=3.5cm,width=3.5cm]{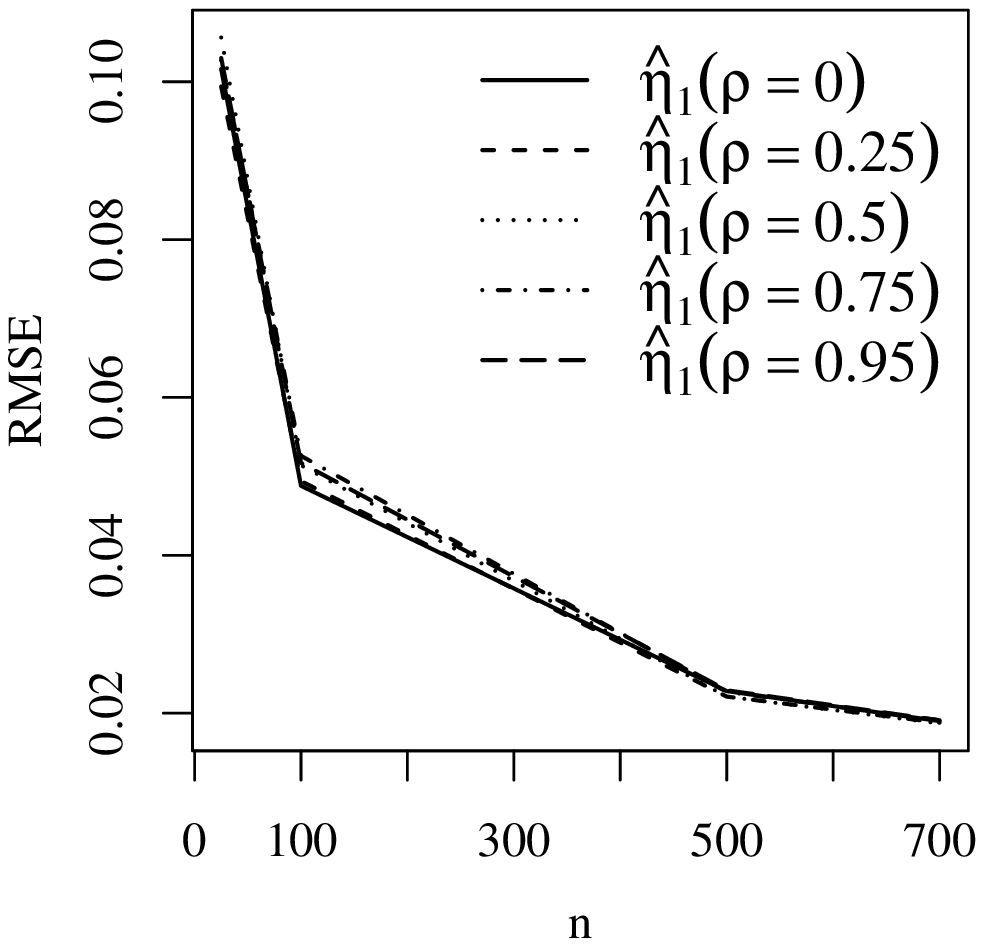}}\hspace{-0.25cm}
{\includegraphics[height=3.5cm,width=3.5cm]{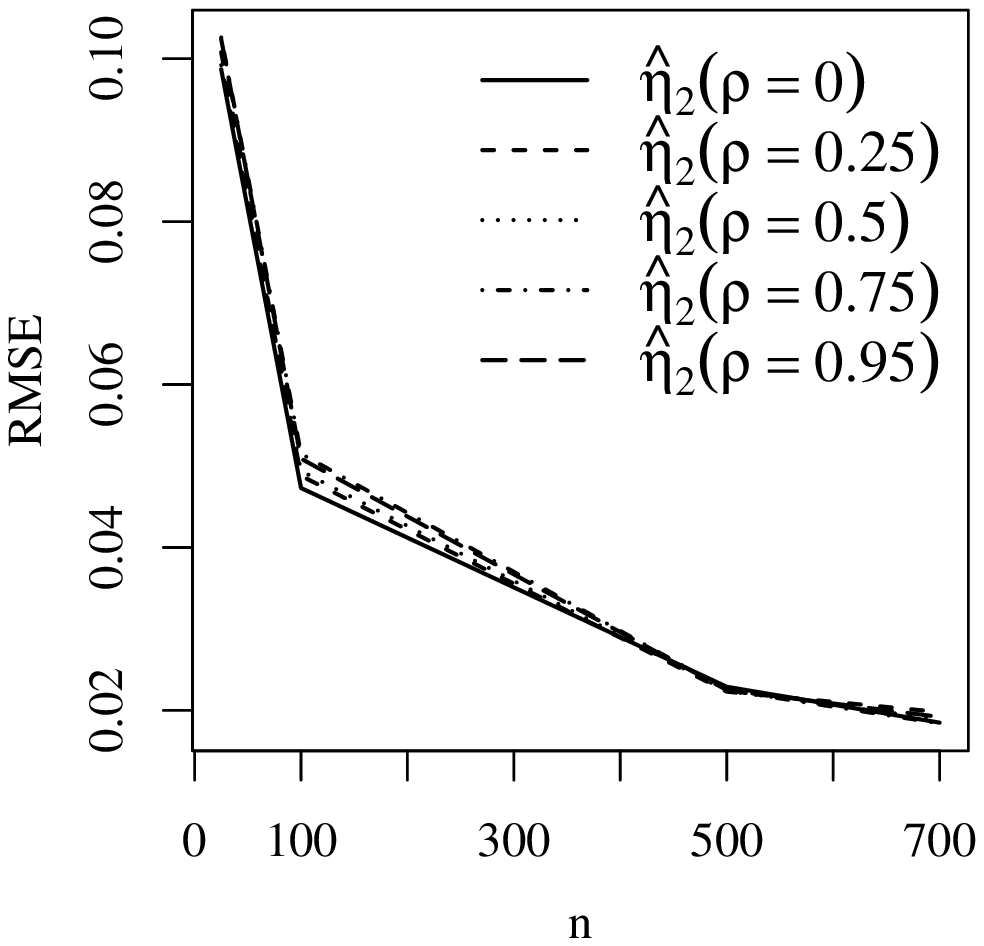}}\hspace{-0.25cm}
{\includegraphics[height=3.5cm,width=3.5cm]{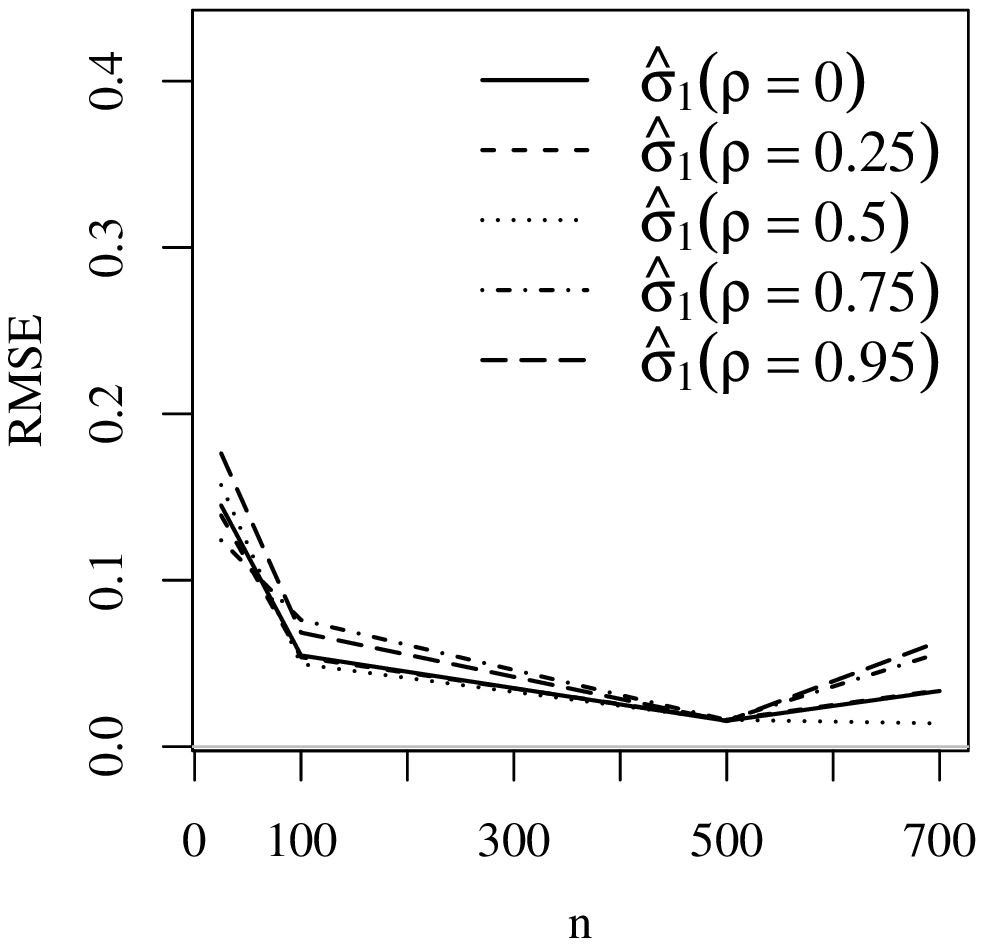}}\hspace{-0.25cm}
{\includegraphics[height=3.5cm,width=3.5cm]{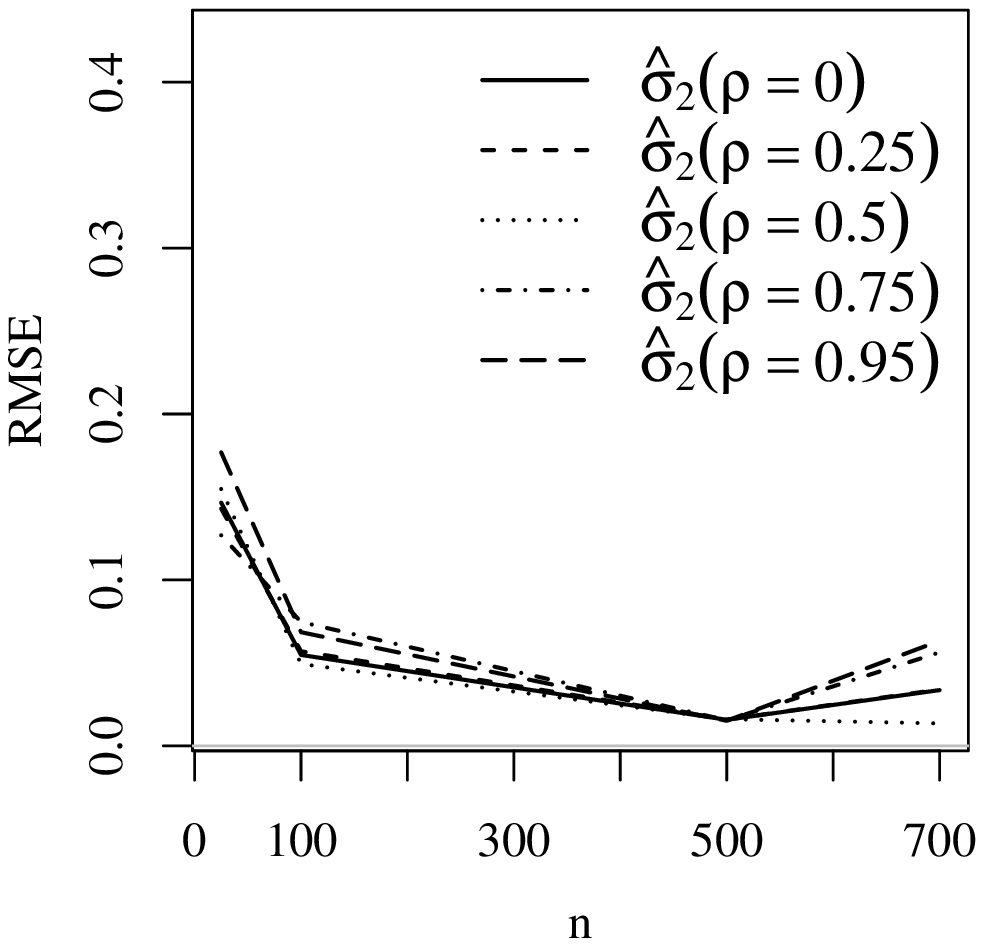}}\hspace{-0.25cm}
{\includegraphics[height=3.5cm,width=3.5cm]{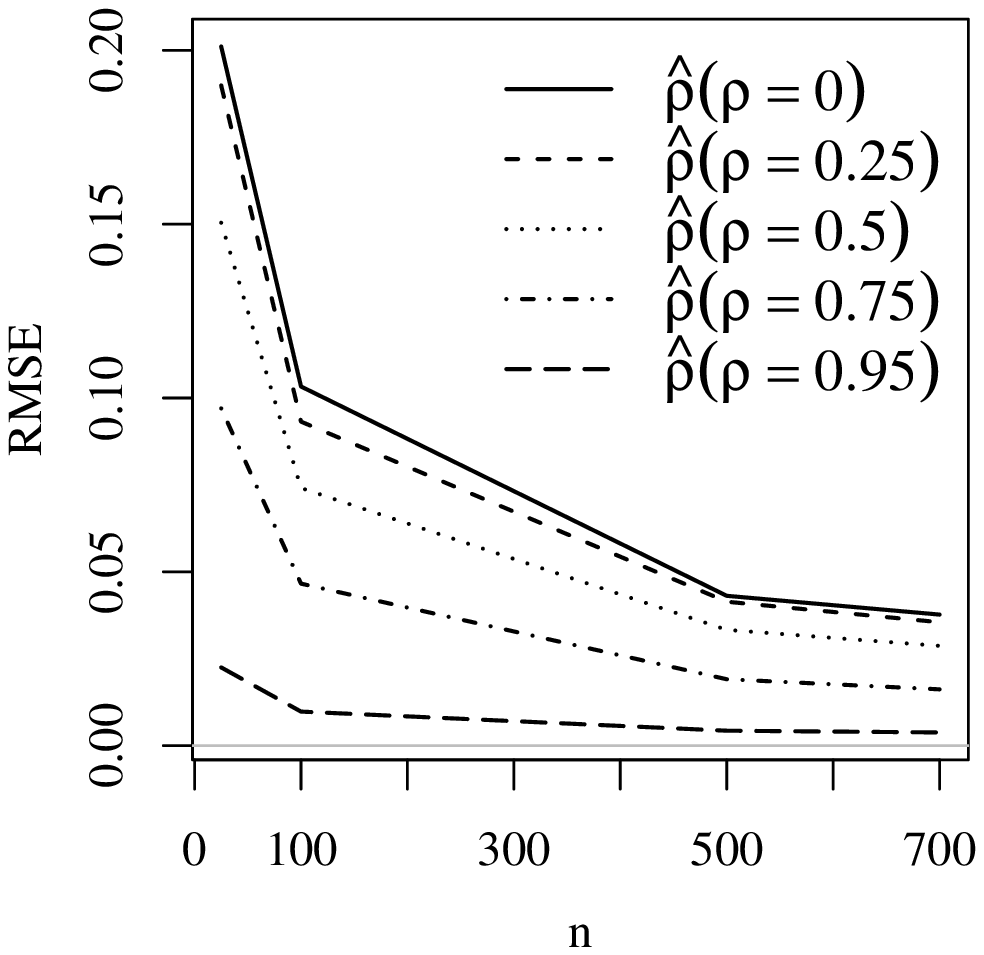}}
{\includegraphics[height=3.5cm,width=3.5cm]{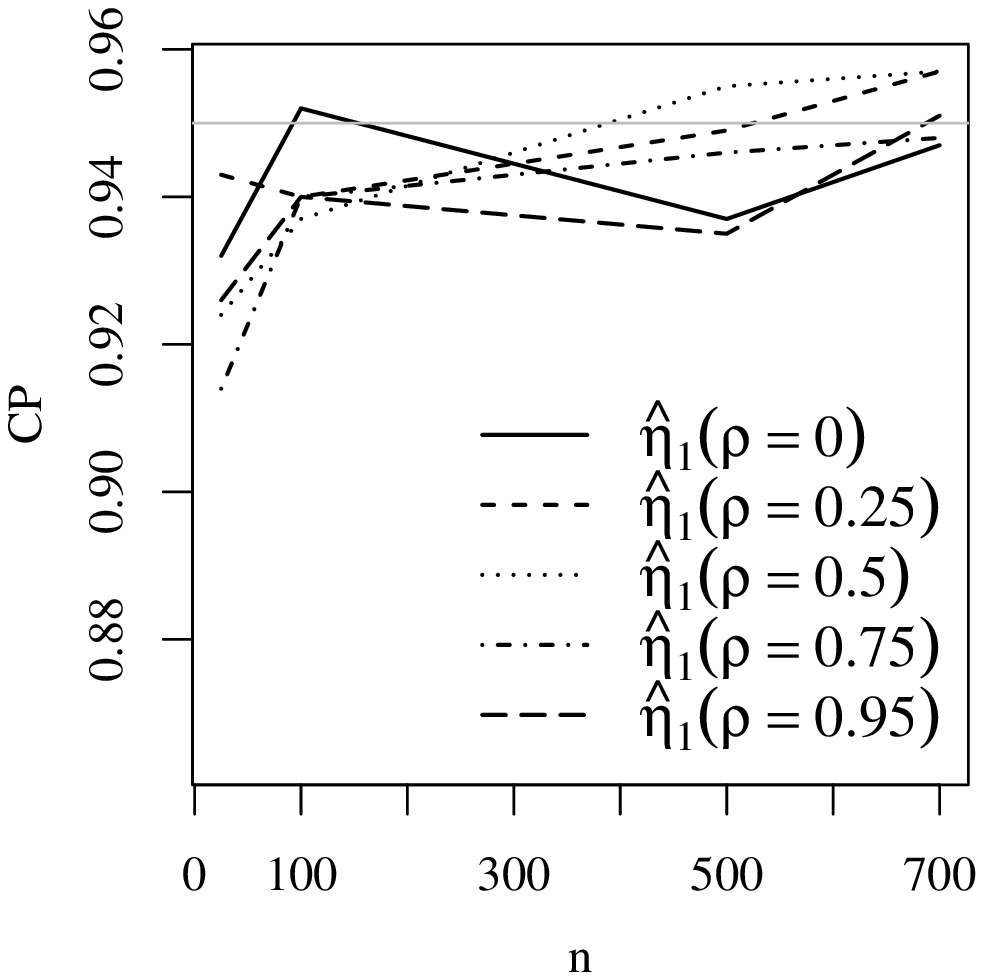}}\hspace{-0.25cm}
{\includegraphics[height=3.5cm,width=3.5cm]{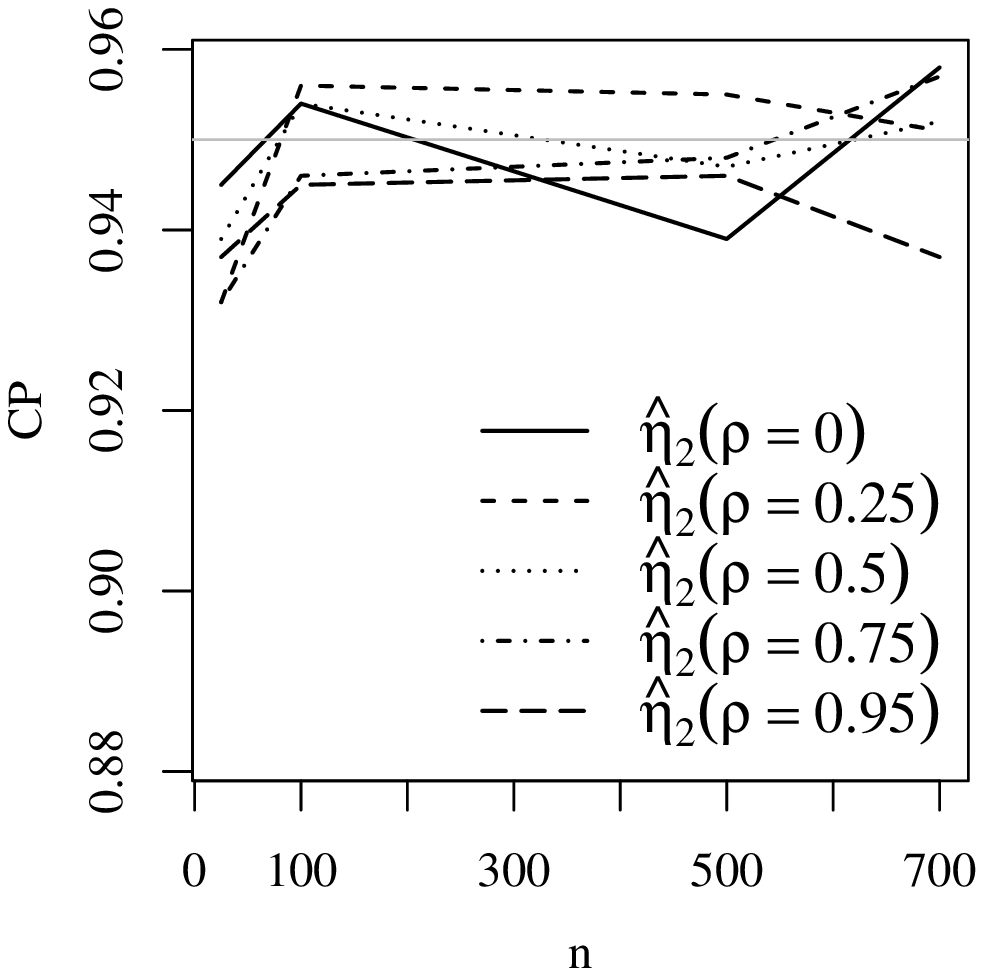}}\hspace{-0.25cm}
{\includegraphics[height=3.5cm,width=3.5cm]{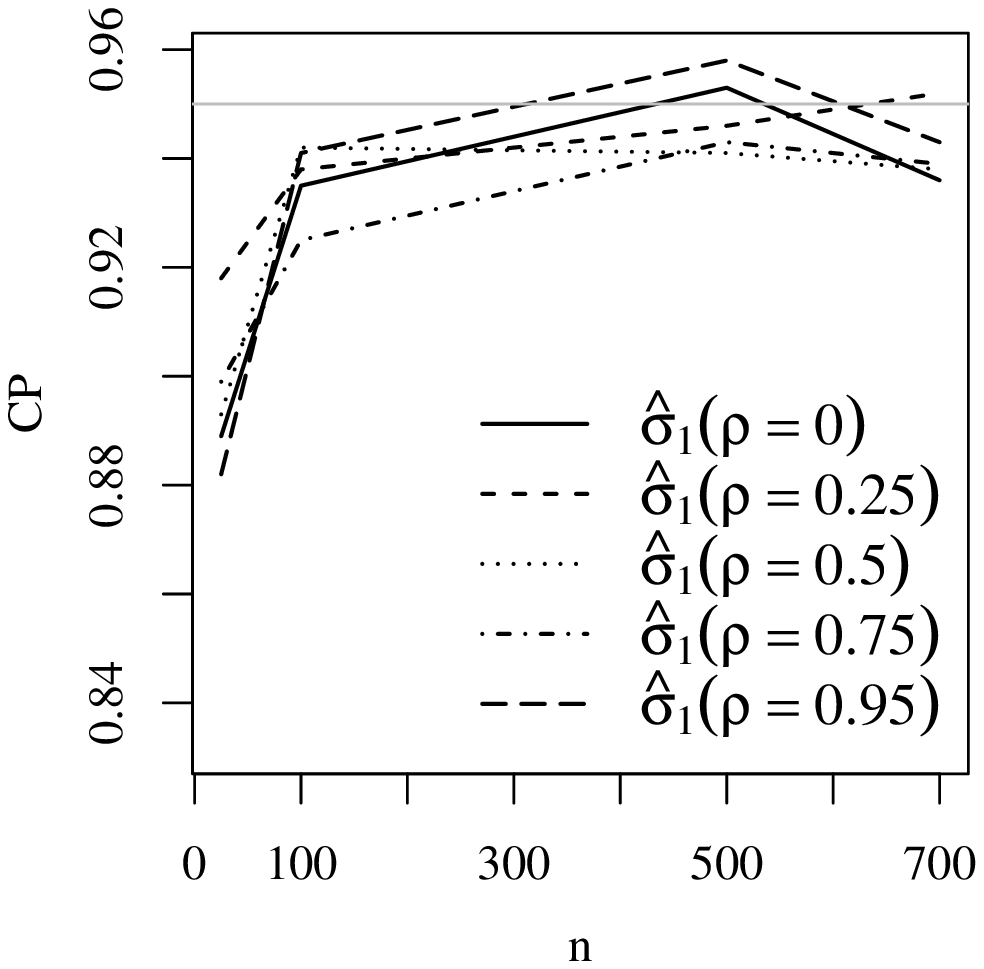}}\hspace{-0.25cm}
{\includegraphics[height=3.5cm,width=3.5cm]{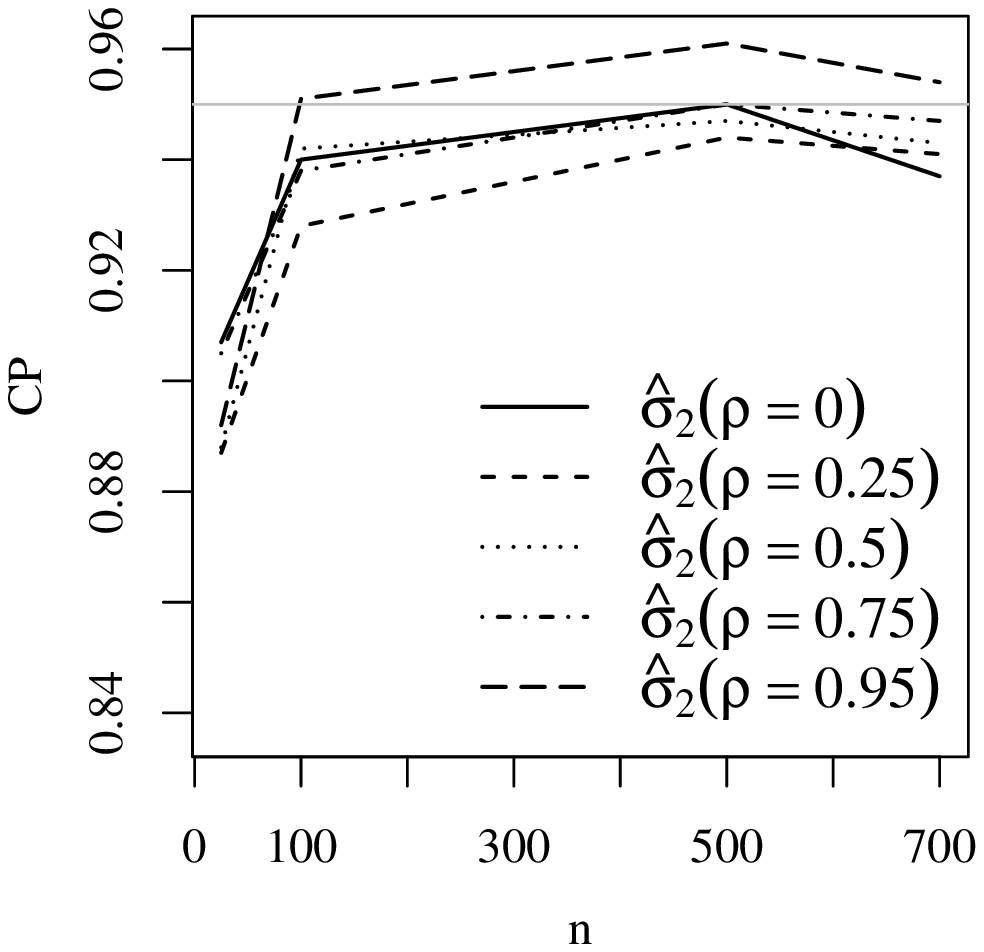}}\hspace{-0.25cm}
{\includegraphics[height=3.5cm,width=3.5cm]{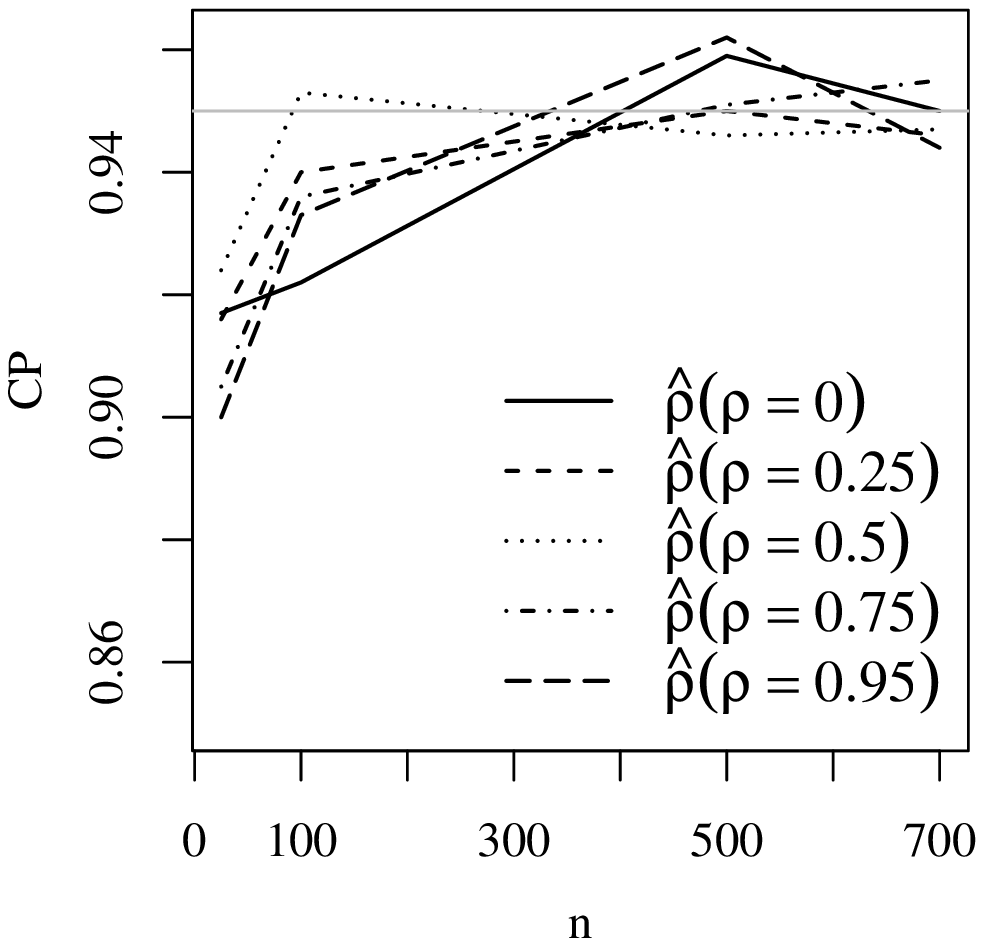}}
\vspace{-0.2cm}
\caption{Monte Carlo simulation results for the bivariate unit-log-normal model.}
\label{fig_normal_mc}
\end{figure}

\section{Application to soccer data} \label{Sec:6}
	\noindent
In this section, two real soccer data sets, corresponding to times elapsed until scored goals of UEFA Champions League and pass completions of 2022 FIFA World Cup, are analyzed. The UEFA Champions League data set was extracted from \cite{Meintanis2007}, whereas the 2022 FIFA World Cup data set is new and is analyzed for the first time here.

\subsection{UEFA Champions League}\label{Sec:6.1}
We consider a bivariate data set on the group stage of the UEFA Champions League for the seasons 2004/05
and 2005/06. Only matches with at least one goal scored directly from a kick by any team, and with at least one goal scored by the home team, are considered; see \cite{Meintanis2007}. The first variable ($W_1$) is the time (in minutes) elapsed until a first kick goal is scored by any team, and the second one $W_2$ is the time (in minutes) elapsed until a first goal of any type is scored by the home team. The times are divided by 90 minutes (full game time) to obtain data on the unit square $(0, 1) \times (0, 1)$; see Table~\ref{table:datasets}.

Table \ref{table:desc} provides descriptive statistics for the variables $W_1$ and $W_2$, including minimum, median, mean, maximum, standard deviation (SD), coefficient of variation (CV), coefficient of skewness (CS), and coefficient of kurtosis (CK). We observe in the variable $W_1$, the mean and median to be, respectively, $0.454$ and $0.456$, i.e., the mean is almost equal to the median, which indicates symmetry in the data. The CV is $49.274\%$, which means a moderate level of dispersion is present around the mean. Furthermore, the CS value also confirms the symmetry nature. The variable $W_2$ has mean to be $0.365$ and median to be $0.311$, which indicates a small positively skewed feature in the distribution of the data. Moreover, the CV value is $69.475\%$, showing a moderate level of dispersion around the mean. The CS confirms the small skewed nature and the CK value indicates the small kurtosis feature in the data.

\begin{table}[H]
\caption{Summary statistics for the UEFA Champions League data set.}
\centering
\begin{tabular}{lccccccccc}
\hline
Variables   & $n$  & Minimum & Median   & Mean   & Maximum & SD     & CV     & CS    & CK     \\ \hline
$W_1$       & 37   & 0.022   & 0.456    & 0.454  & 0.911   & 0.224  & 49.274 & 0.164 & -0.930 \\
$W_2$       & 37   & 0.022   & 0.311    & 0.365  & 0.944   & 0.254  & 69.475 & 0.522 & -0.839  \\ \hline
\end{tabular}
\label{table:desc}
\end{table}

The ML estimates and the standard errors (in parentheses) for the bivariate unit-log-symmetric model parameters are presented in
Table \ref{table:est}. The extra parameters, associated with log-Student-$t$, log-hyperbolic and log-slash models, were estimated by using the profile log-likelihood; see \cite{ssls:22}. Table \ref{table:est} also presents the log-likelihood value, and the values of the Akaike (AIC) and Bayesian (BIC) information criteria. We observe that the log-hyperbolic model provides better fit than other models based on the values of log-likelihood, AIC and BIC.
Note, however, that the values of log-likelihood, AIC and BIC of all bivariate unit-log-symmetric models are quite close to each other.

\begin{table}[H]
\caption{ML estimates (with standard errors in parentheses), and log-likelihood, AIC and BIC values for the indicated bivariate unit-log-symmetric models.}
\resizebox{\linewidth}{!}{
\begin{tabular}{lccccccccc}
\noalign{\hrule height 1.7pt}
Distribuiton & $\widehat{\eta}_1$  & $\widehat{\eta}_2$ & $\widehat{\sigma}_1$ & $\widehat{\sigma}_2$ & $\widehat{\rho}$ & $\widehat{\nu}$ & Log-likelihood & AIC &   BIC  \\ \hline
Log-normal         & 0.5288*     & 0.3414*     & 0.8865*   & 1.1355*     &  0.4956*  & --     & -36.693   & 83.386  & 91.441          \\
				   & (0.0771)    & (0.0637)    & (0.1031)  & (0.1320)    & (0.1240)  &        &           &         & \\
Log-Student-$t$    & 0.5541*     & 0.3783*     & 0.7431*   & 0.9734*     &  0.4723*  & 7      & -35.487   & 80.974  & 89.029          \\
				   & (0.0751)    & (0.0672)    & (0.1033)  & (0.1308)    & (0.1463)  &        &           &         & \\
Log-hyperbolic     & 0.5458*     & 0.3816*     & 0.8456*   & 1.0950*     &  0.4893*  & 2     & -35.470    & 80.940  & 88.996          \\
				   & (0.0752)    & (0.0677)    & (0.1162)  & (0.1462)    & (0.1428)  &        &           &         & \\
Log-Laplace        & 0.5680*     & 0.5679*     & 0.9928*   & 1.3231*     &  0.5281*  & --     & -36.009   & 82.019  & 90.073          \\
				   & (0.0020)    & (0.0021)    & (0.1692)  & (0.2164)    & (0.1639)  &        &           &         & \\
Log-slash          & 0.5629*     & 0.3715*     & 0.6203*   & 0.8302*     &  0.4472*  & 5      & -35.560   & 81.120  & 89.174          \\
				   & (0.0749)    & (0.0666)    & (0.0847)  & (0.1096)    & (0.1472)  &        &           &         & \\        \\\hline
\end{tabular}
}
\label{table:est}
\footnotesize{$^*$ significant at 5\% level.}
\end{table}

Figure~\ref{fig:qqplots} shows the QQ plots of the Mahalanobis distance for the bivariate unit-log-symmetric models considered in Table \ref{table:est}. The QQ plot is a plot of the empirical quantiles of the Mahalanobis distance against the theoretical quantiles of the respective reference distribution (see Section \ref{maha_sec}). Therefore, points falling along a straight line would indicate a good fit. From Figure~\ref{fig:qqplots}, we see clearly that, with the exception of log-Student-$t$ case, the Mahalanobis distances in the considered models conform relatively well with their reference distributions. We also see that, in all the cases, there is a point away from the reference line, which may be an outlier.

\begin{figure}[!ht]
\centering
\subfigure[Log-normal]{\includegraphics[height=5cm,width=5cm]{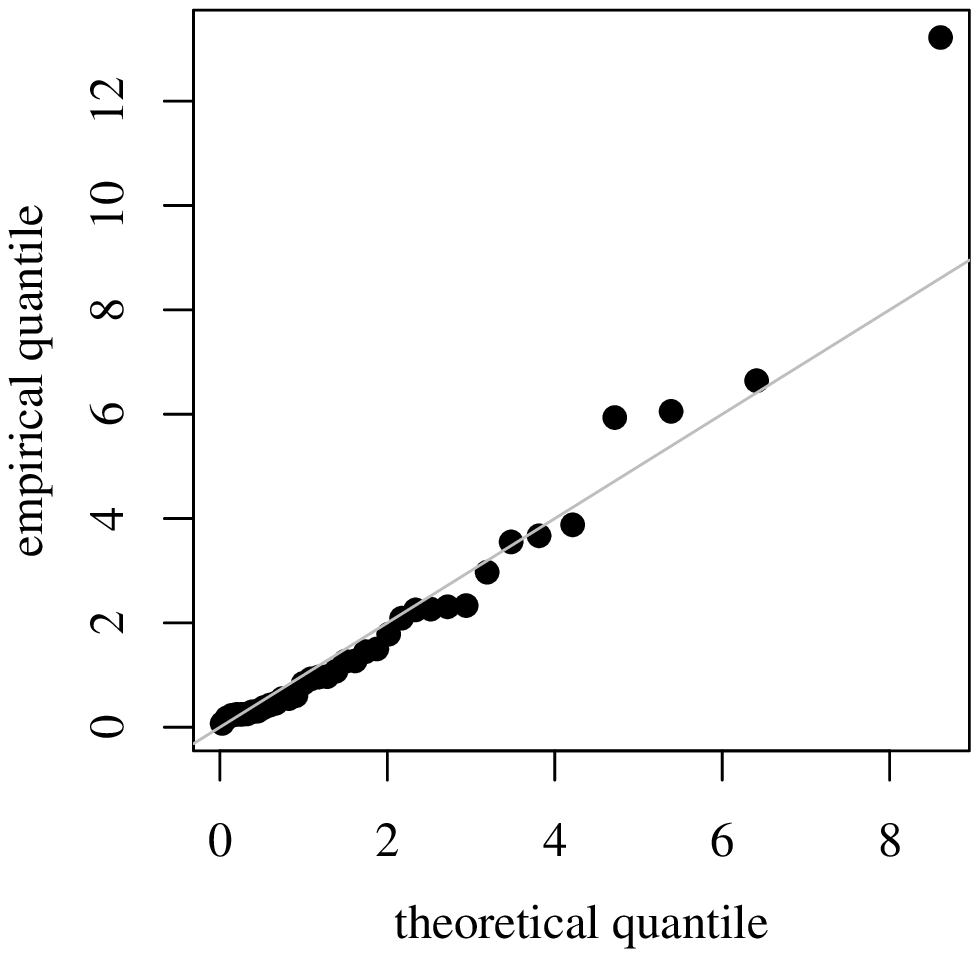}}
\subfigure[Log-Student-$t$]{\includegraphics[height=5cm,width=5cm]{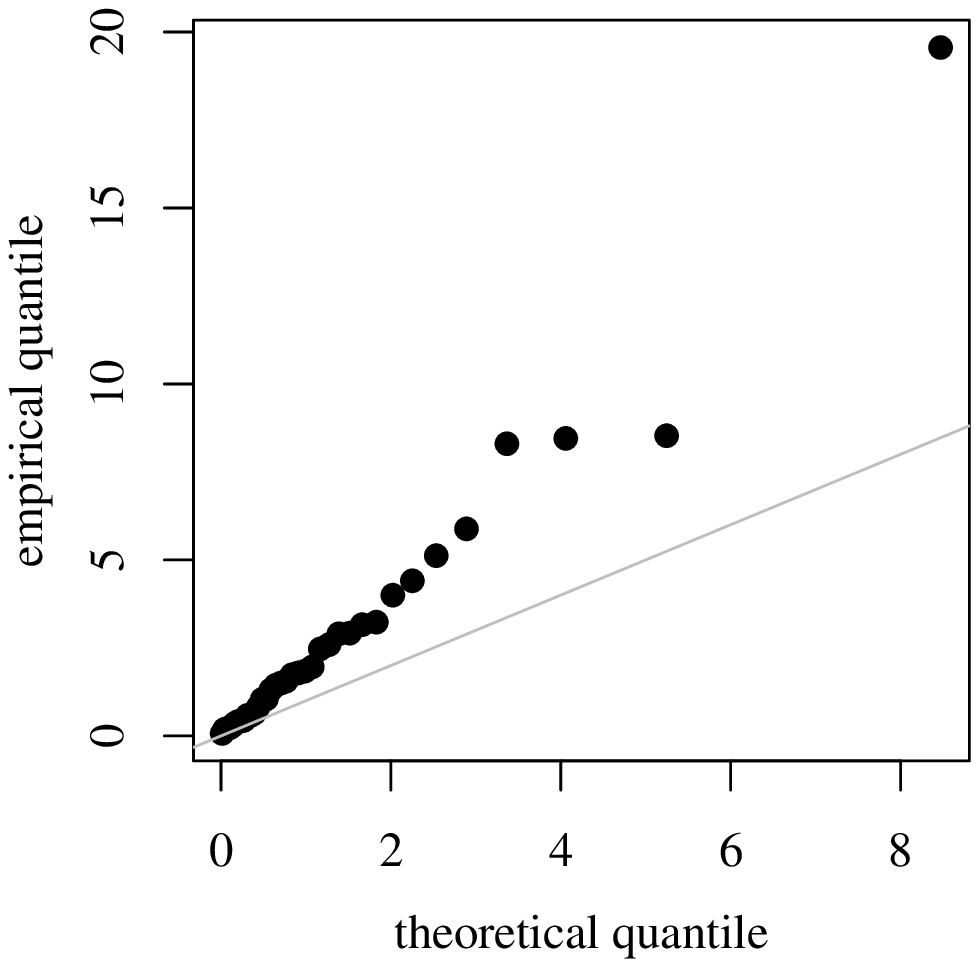}}
\subfigure[Log-hyperbolic]{\includegraphics[height=5cm,width=5cm]{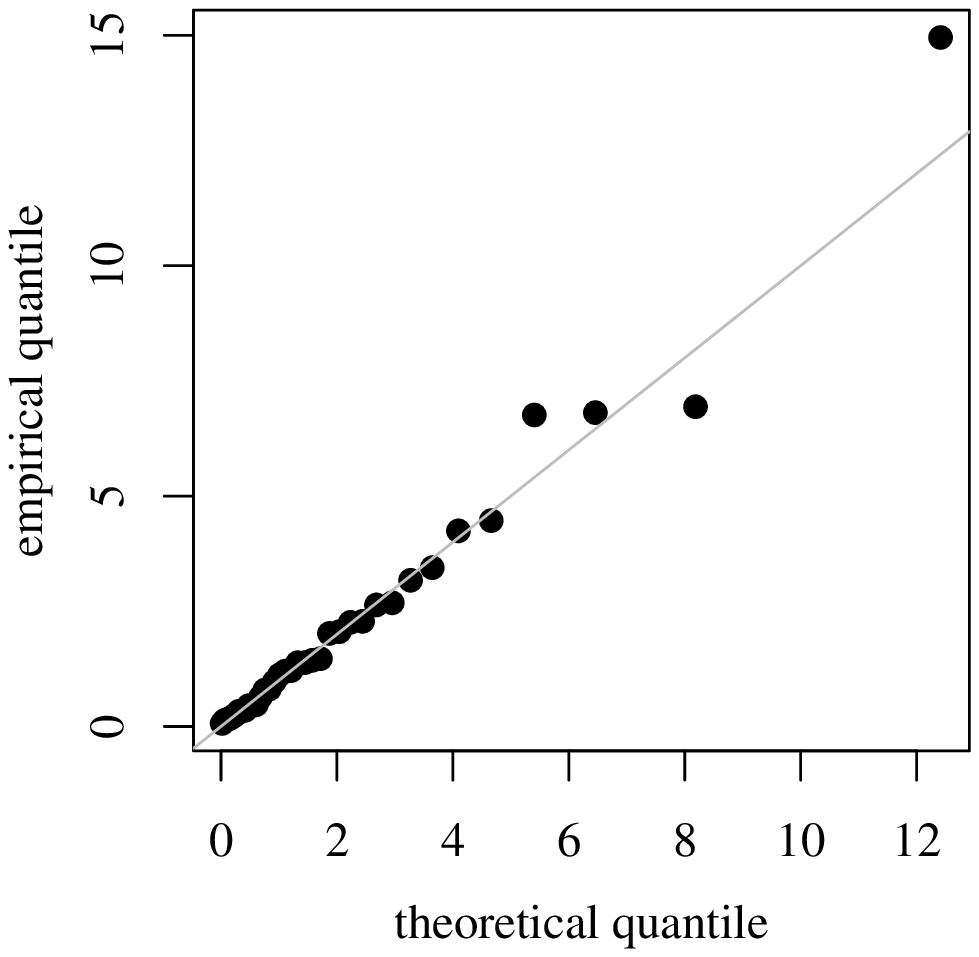}}
\subfigure[Log-Laplace]{\includegraphics[height=5cm,width=5cm]{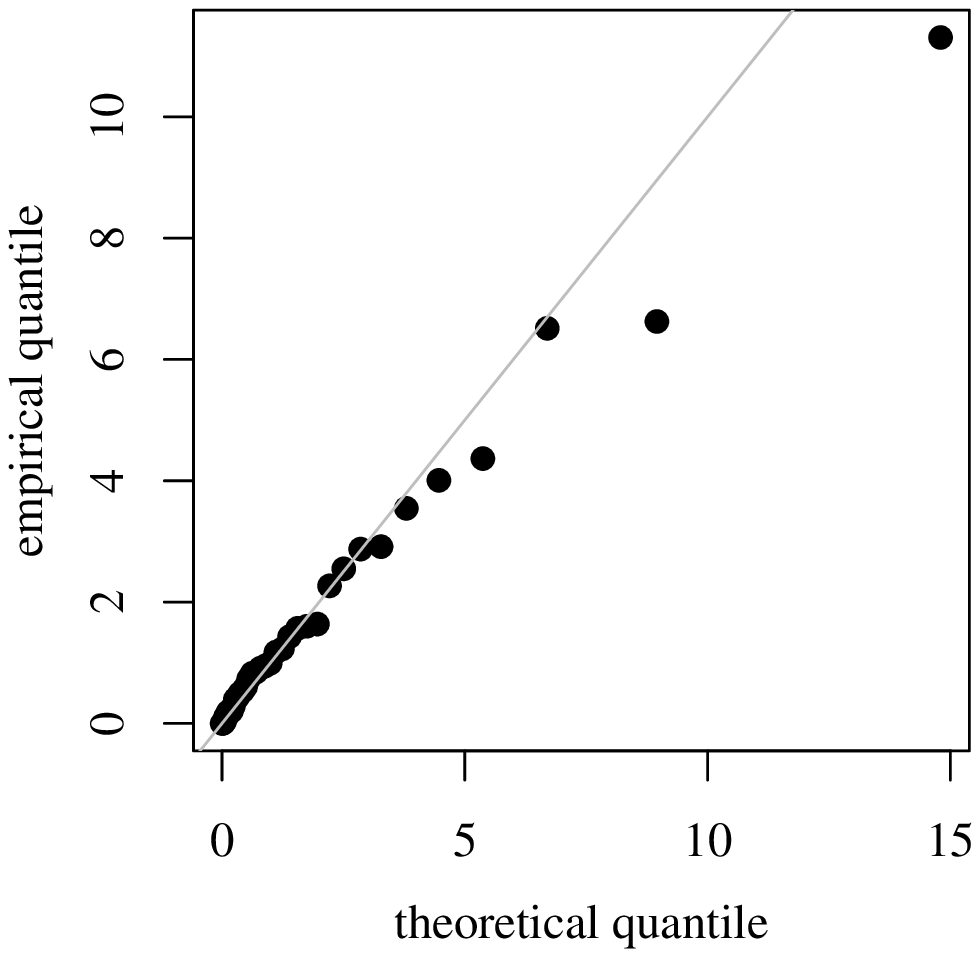}}
\subfigure[Log-slash]{\includegraphics[height=5cm,width=5cm]{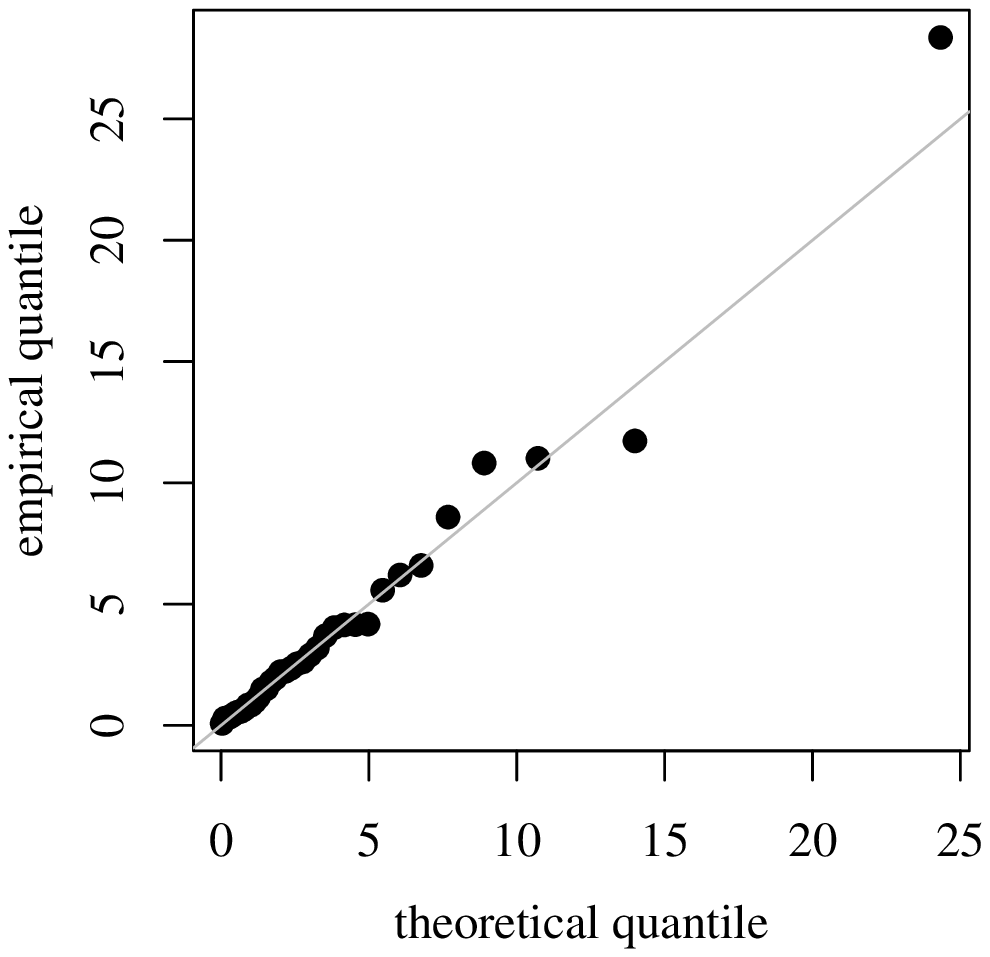}}
 \caption{\small {QQ plot of Mahalanobis distances for the indicated models.}}
\label{fig:qqplots}
\end{figure}


\subsection{2022 FIFA World Cup}\label{Sec:6.2}

We now use the data on the 2022 FIFA World Cup to illustrate the model developed in the preceding sections. The data are available at \url{https://www.kaggle.com/}.  The first variable ($W_1$) is the medium pass completion proportion, that is, successful passes between 14 and 18 meters. The second variable ($W_2$) is the long pass completion proportion, namely, passes longer than 37 meters; see Table~\ref{table:datasets}.

Table \ref{table:desc2} provides descriptive statistics for the variables $W_1$ and $W_2$. We observe in the variable $W_1$, the mean and median to be, respectively, $0.454$ and $0.456$, i.e., the mean is almost equal to the median, which indicates symmetry in the data. The CV is $49.274\%$, which means a moderate level of dispersion around the mean. Furthermore, the CS value also confirms the symmetry nature. The variable $W_2$ has mean equal to $0.365$ and median equal to $0.311$, which indicates a small positively skewed feature in the distribution of the data. Moreover, the CV value is $69.475\%$, showing a moderate level of dispersion around the mean. The CS confirms the small skewed nature and the CK value indicates the small kurtosis feature in the data.

\begin{table}[H]
\caption{Summary statistics for the 2022 FIFA World Cup data set.}
\centering
\begin{tabular}{lccccccccc}
\hline
Variables   & $n$  & Minimum & Median   & Mean   & Maximum & SD     & CV     & CS     & CK     \\ \hline
$W_1$       & 32   & 0.769   & 0.860    & 0.860  & 0.931   & 0.038  & 4.376  & -0.373 & -0.194 \\
$W_2$       & 32   & 0.427   & 0.556    & 0.550  & 0.751   & 0.075  & 13.713 & 0.308  & -0.425  \\ \hline
\end{tabular}
\label{table:desc2}
\end{table}

Table \ref{table:est2} presents the estimation results for the bivariate unit-log-symmetric models, and these reveal that the log-normal model provides better fit than all other models based on the values of log-likelihood, AIC and BIC.

Figure~\ref{fig:qqplots2} shows the QQ plots of the Mahalanobis distances (see Section \ref{maha_sec}) for the bivariate unit-log-symmetric models considered in Table \ref{table:est2}. We see clearly that the log-normal model provides better fit than all other bivariate unit-log-symmetric models.

\begin{table}[H]
\caption{ML estimates (with standard errors in parentheses), and log-likelihood, AIC and BIC values for the indicated bivariate unit-log-symmetric models.}
\resizebox{\linewidth}{!}{
\begin{tabular}{lccccccccc}
\noalign{\hrule height 1.7pt}
Distribuiton & $\widehat{\eta}_1$  & $\widehat{\eta}_2$ & $\widehat{\sigma}_1$ & $\widehat{\sigma}_2$ & $\widehat{\rho}$ & $\widehat{\nu}$ & Log-likelihood  & AIC &   BIC  \\ \hline
Log-normal         & 1.9872*     & 0.7953*     & 0.1364*   & 0.2089*     &  0.7343*  &   --   & 20.791    &-31.581  &-24.252          \\
				   & (0.0479)    & (0.0294)    & (0.0171)  & (0.0261)    & (0.0815)  &        &           &         & \\
Log-Student-$t$    & 1.9954*     & 0.7936*     &-0.1257*   &-0.1949*     &  0.7423*  & 9      & 20.130    &-30.260  &-22.931          \\
				   & (0.0485)    & (0.0299)    & (0.0178)  & (0.0271)    & (0.0868)  &        &           &         & \\
Log-hyperbolic     & 1.9908*     & 0.7942*     & 0.3956*   & 0.6088*     &  0.7378*  & 10     & 20.618    &-31.236  &-23.907          \\
				   & (0.0482)    & (0.0296)    & (0.0523)  & (0.0800)    & (0.0841)  &        &           &         & \\
Log-Laplace        & 1.9908*     & 0.8089*     & 0.1563*   & 0.2425*     &  0.7471*  & --     & 16.915   &-23.830  &-16.501          \\
				   & (0.0023)    & (0.0021)    & (0.0278)  & (0.0415)    & (0.0938)  &        &           &         & \\
Log-slash          & 1.9897*     & 0.7935*     & 0.1173*   & 0.1802*     &  0.7392*  & 8      & 20.613   &-28.919  &-23.898          \\
				   & (0.0482)    & (0.0295)    & (0.0154)  & (0.0237)    & (0.0844)  &        &           &         & \\        \\\hline
\end{tabular}
}
\label{table:est2}
\footnotesize{$^*$ significant at 5\% level.}
\end{table}

\begin{figure}[!ht]
\centering
\subfigure[Log-normal]{\includegraphics[height=5cm,width=5cm]{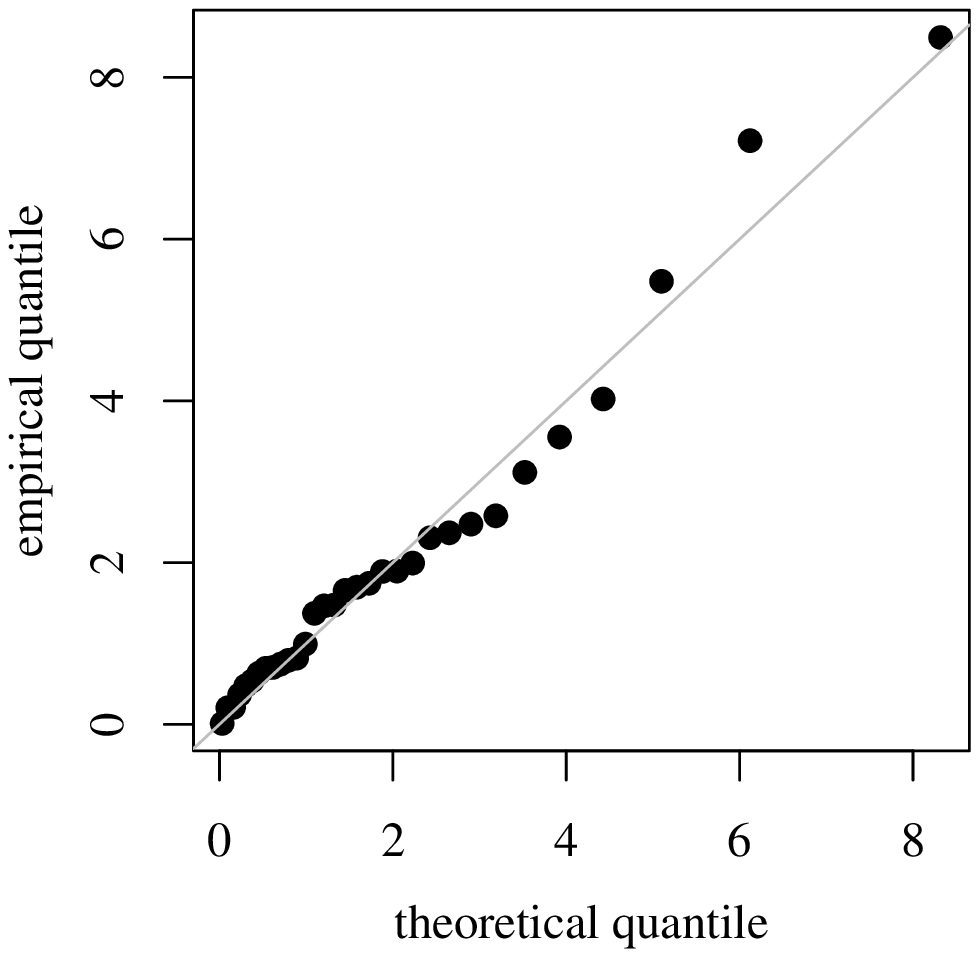}}
\subfigure[Log-Student-$t$]{\includegraphics[height=5cm,width=5cm]{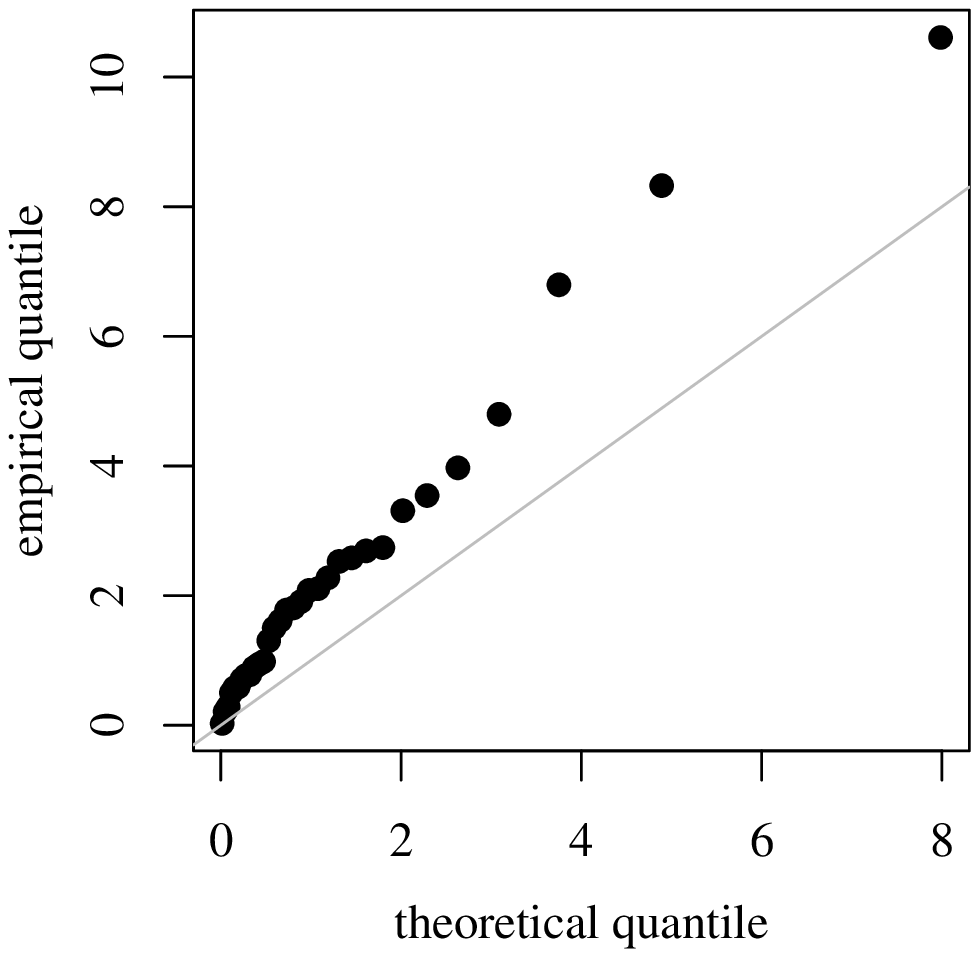}}
\subfigure[Log-hyperbolic]{\includegraphics[height=5cm,width=5cm]{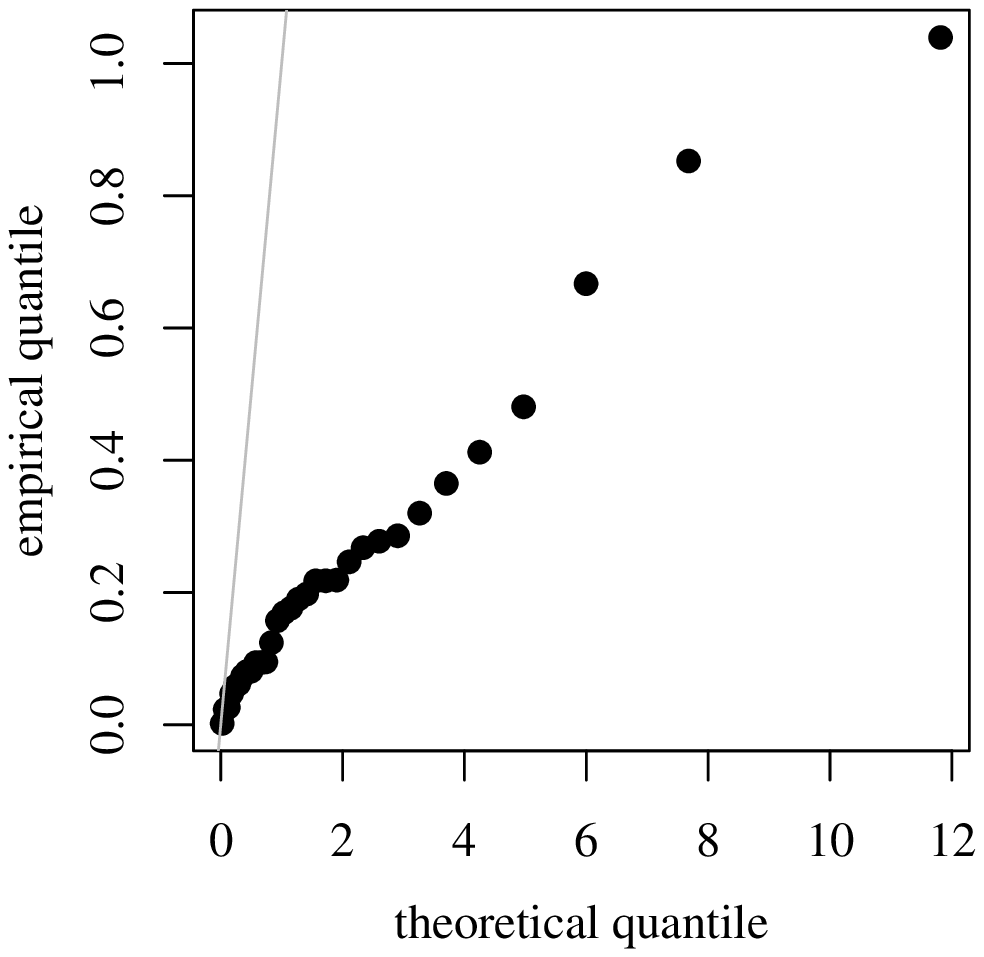}}
\subfigure[Log-Laplace]{\includegraphics[height=5cm,width=5cm]{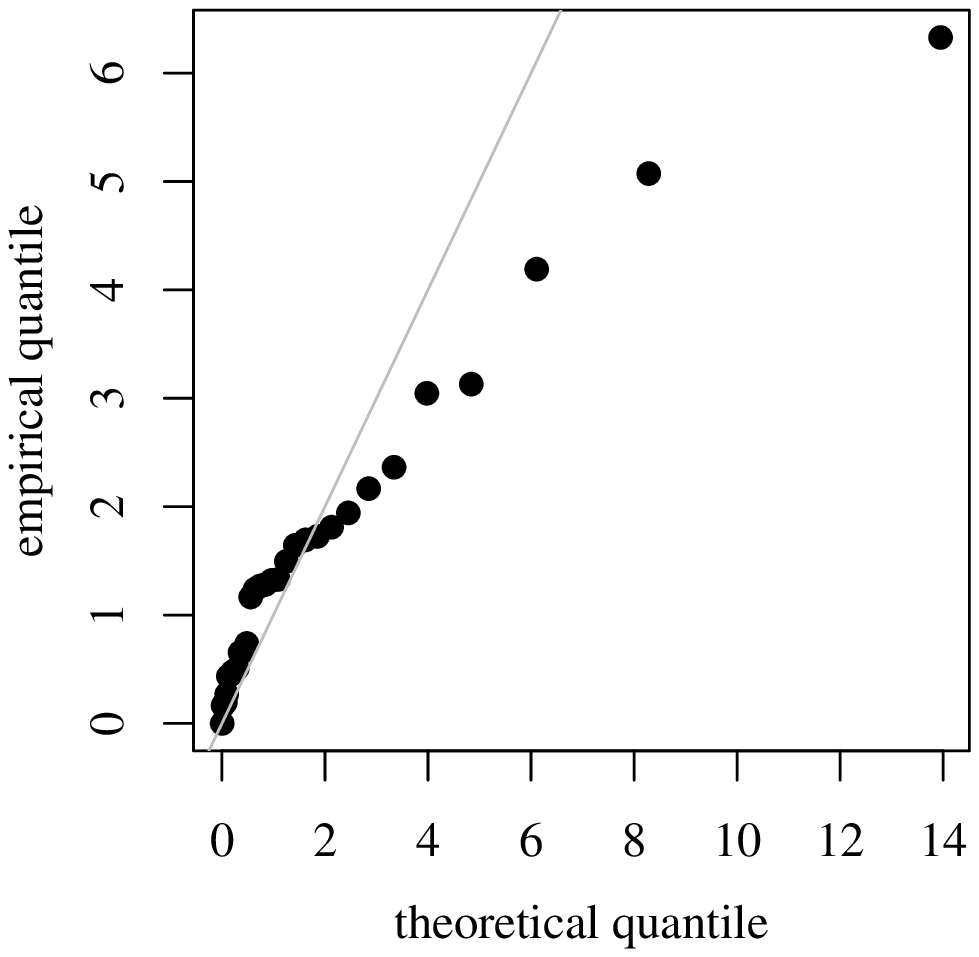}}
\subfigure[Log-slash]{\includegraphics[height=5cm,width=5cm]{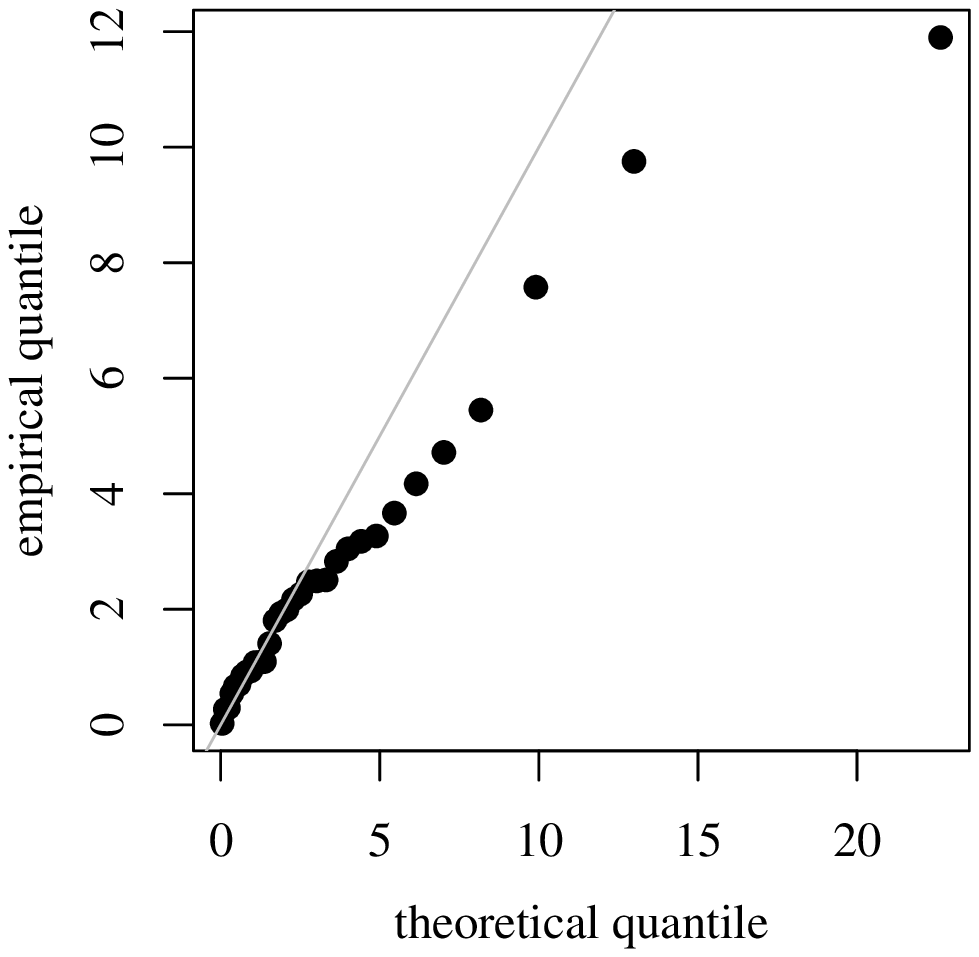}}
 \caption{\small {QQ plot of Mahalanobis distances for the indicated models.}}
\label{fig:qqplots2}
\end{figure}


\section{Conclusions} \label{Sec:7}
\noindent
In this paper, we have proposed a family of bivariate distributions over the unit square. By suitably defining the density generator, we can transform any distribution over the real line into a bivariate distribution over the region
$(0,1)\times(0,1)$. Such a model has several potential applications, since the simultaneous modeling of quantities like proportions, rates or indices frequently arises in applied sciences like economics, medicine, engineering and social sciences. We have discussed several theoretical properties like stochastic representation, quantiles, conditional distributions, independence and moments. We have also carried out a Monte Carlo simulation study and finally demonstrated some applications to soccer data. The present research can be extended in several possible directions. By changing the density generator, numerous special forms of the BULS distribution can be constructed. Furthermore, generalizations to higher dimensions can be studied. We are currently working in these directions and hope to report the findings in a future paper.

	\paragraph{Acknowledgements}
	Roberto Vila and Helton Saulo gratefully acknowledge financial support from CNPq, CAPES and FAP-DF, Brazil.
	
	\paragraph{Disclosure statement}
	There are no conflicts of interest to disclose.



\newpage
\begin{appendices}
\section{Some additional results}
For the convenience of readers, we present here some complementary results relating to Section \ref{cond-dist}.

\begin{definition}\label{def-GH}
We say that a random variable $X$ follows a univariate generalized hyperbolic  (GH) distribution, denoted by $X\sim {\rm GH}(\lambda,\alpha,\delta)$, if its PDF is given by
\begin{align*}
f_{\rm GH}(x;\lambda,\alpha,\delta)
=
\dfrac{(\sqrt{\alpha^2}/\delta)^\lambda}{\sqrt{2\pi} K_\lambda(\delta \sqrt{\alpha^2}\,)}\, 
\,\dfrac{K_{\lambda-1/2}(\alpha\sqrt{\delta^2+x^2})}{(\sqrt{\delta^2+x^2}/\alpha)^{1/2-\lambda}},
\quad -\infty<x<\infty.
\end{align*}
Here, $K_r$ is the modified Bessel function
of the third kind with index $r$, $\lambda\in\mathbb{R}, \alpha\in\mathbb{R}$ and 
$\delta>0$ is a  scale parameter.
\end{definition}

The following result has appeared in a multivariated version in  \citet[][Proposition 3, p. 5]{dy:18}.
\begin{proposition}[Hyperbolic generator]\label{cond-marg-Hyp}
	Let  $\boldsymbol{Z}=(Z_1,Z_2)^\top$ be a random vector as in  Proposition \ref{Stochastic Representation}.
	If $g_c(x)=\exp(-\nu\sqrt{1+x}\,)$, then
	the conditional distribution of $Z_{2}$, given $Z_{1}=x$, is ${\rm GH}(1,\nu,\sqrt{1+x^2}\,)$ and both of its unconditional distributions are ${\rm GH}(3/2,\nu,1)$.
\end{proposition}
\begin{proof}
	By using \eqref{pdf-Zs}, the joint PDF of $Z_1$ and $Z_2$ is (see Table \ref{table:1})  
	\begin{align}\label{eq-app0}
	f_{Z_1,Z_2}(x,y)={\nu^2\exp(\nu)\over 2\pi (\nu+1)}\,
	\exp\left(-\nu\sqrt{1+x^2+y^2}\,\right).
	\end{align}
	
	So, the marginal PDF of $Z_1$ is given by
\begin{align*}
f_{Z_1}(x)
=
\int_{-\infty}^{\infty}
f_{Z_1,Z_2}(x,y) \,{\rm d}y
&=
{\nu^2\exp(\nu)\over 2\pi (\nu+1)}\,
\int_{-\infty}^{\infty}
\exp\left(-\nu\sqrt{1+x^2+y^2}\,\right) \,{\rm d}y
\nonumber
\\[0,2cm]
&=
{2\nu^2\exp(\nu)\over 2\pi (\nu+1)}\,
\int_{0}^{\infty}
\exp\left(-\nu\sqrt{1+x^2+y^2}\,\right) \,{\rm d}y.
\end{align*}
By using Formula 6 of Section 3.46-3.48 of \citet[][p. 364]{Gradshteyn2000} that $\int_{0}^{\infty}\exp(-a\sqrt{b^2+x^2}){\rm d}x=b K_1(ab)$, the above integral is
\begin{align}\label{form-interm}
=
{2\nu^2\exp(\nu)\over 2\pi (\nu+1)}\, \sqrt{1+x^2} K_1(\nu\sqrt{1+x^2}\,).
\end{align}
Now, as $K_{3/2}(\nu)=2\sqrt{\pi}\exp(-\nu) (\nu+1)(\nu/2)^{3/2}/\nu^3$, the above espression becomes
\begin{align*}
	=
	{\nu^{3/2}\over\sqrt{2\pi}K_{3/2}(\nu)}\, {K_1(\nu\sqrt{1+x^2})\over (\sqrt{1+x^2}/\nu\,)^{-1}}
	=
	f_{\rm GH}(x;3/2,\nu,1)
\end{align*}
which proves that  $Z_1\sim{\rm GH}(3/2,\nu,1)$. Similarly, we can show that $Z_2\sim{\rm GH}(3/2,\nu,1)$, as well.

On the other hand, from \eqref{eq-app0} and \eqref{form-interm}, and by using the well-known identity $K_{1/2}(z)=\sqrt{\pi/(2z)} \exp(-z)$, the conditional PDF of $Z_2$, given $Z_1=x$, is obtained as
\begin{align*}
	f_{Z_2\,\vert\, Z_1}(y\,\vert\,x)
	=
	\dfrac{
		\exp(-\nu\sqrt{1+x^2+y^2}\,)}{2\sqrt{1+x^2} K_1(\nu\sqrt{1+x^2}\,)}
	&=
	\dfrac{\nu/\sqrt{1+x^2} }{\sqrt{2\pi} K_1(\nu\sqrt{1+x^2}\,)}\, 
		\dfrac{
		K_{1/2}(\nu\sqrt{1+x^2+y^2}\,)}{(\sqrt{1+x^2+y^2}/\nu)^{-1/2}}
	\\[0,2cm]
	&=
	f_{\rm GH}(x;1,\nu,\sqrt{1+x^2}),
\end{align*}
which complete the proof.
\end{proof}

The following result has also appeared in a multivariated version in \citet[][Theorem 6.7.1, p. 253]{Kotz2001}.
\begin{proposition}[Laplace generator]\label{cond-marg-Laplace}
Let  $\boldsymbol{Z}=(Z_1,Z_2)^\top$ be a random vector as in  Proposition \ref{Stochastic Representation}.
If $g_c(x)=K_0(\sqrt{2x}\,)$, then
the conditional distribution of $Z_{2}$, given $Z_{1}=x$, is ${\rm GH}(1/2,\sqrt{2},\vert x\vert)$ and both of its unconditional distributions are ${\rm Laplace}(0,1/\sqrt{2})$.
\end{proposition}
\begin{proof}
	By \eqref{pdf-Zs} and using the definitions of $K_0$ and $Z_{g_c}$ in Table \ref{table:1}, we have
\begin{align}\label{eq-app1}
	f_{Z_1,Z_2}(x,y)
	=
	{1\over \pi}\, K_0\left(\sqrt{2(x^2+y^2)}\,\right)
	=
	{1\over 2\pi}\,
	\int_0^\infty {1\over t} \exp\left(-t-{x^2+y^2\over 2t}\right) {\rm d}t.
\end{align}

We then find the marginal density of $Z_1$ to be
\begin{align}\label{eq-app2}
	f_{Z_1}(x)
	=
\int_{-\infty}^{\infty}
f_{Z_1,Z_2}(x,y) \,{\rm d}y
		=
		{1\over\sqrt{2}}\,\exp(-\sqrt{2}\, \vert x\vert)
		=
		f_{\rm L}(x),
\end{align}
where $f_{\rm L}(x)=\exp(-\sqrt{2}\, \vert x\vert)/\sqrt{2}$ is the Laplace PDF with scale parameter $1/\sqrt{2}$;
that is, $Z_1\sim{\rm Laplace}(0,1/\sqrt{2})$. Similarly, we can show that $Z_2\sim{\rm Laplace}(0,1/\sqrt{2})$, as well.

On the other hand, by using \eqref{eq-app1}, \eqref{eq-app2} and the well-known identity $K_{1/2}(z)=\sqrt{\pi/(2z)} \exp(-z)$, the conditional PDF of $Z_2$, given $Z_1=x$, is obtained as
\begin{align*}
	f_{Z_2\,\vert\, Z_1}(y\,\vert\,x)
	=
	\dfrac{\displaystyle
			{1\over \pi}\,
	 K_0\left(\sqrt{2(x^2+y^2)}\,\right)
	}
{\displaystyle
	 {1\over \sqrt{2}}\exp(-\sqrt{2}\, \vert x\vert)
 }
&=
\dfrac
{
(\sqrt{2}/\vert x\vert)^{1/2}
}
{
\sqrt{2\pi} K_{1/2}(\sqrt{2} \vert x\vert)
}
\, K_0\left(\sqrt{2} \sqrt{x^2+y^2}\,\right)
\\[0,2cm]
&=
f_{\rm GH}(x;1/2,\sqrt{2},\vert x\vert).
\end{align*}
Then, from Definition \ref{def-GH}, we simply have $Z_2\,\vert\, (Z_1=x)\sim {\rm GH}(1/2,\sqrt{2},\vert x\vert)$.
\end{proof}

\begin{definition}\label{def-GH-1}
	We say that a random variable $X$ follows an univariate extended slash (ESL) distribution, denoted by $X\sim {\rm ESL}(a,q)$, if its PDF is given by
	\begin{align*}
	f_{\rm ESL}(x;a,q)=
	\dfrac{\displaystyle\int_{0}^1 t^q \phi(t a)\phi(t x)\, {\rm d}t}
	{\displaystyle \int_{0}^1 u^{q-1} \phi(u a)\, {\rm d}u},
	\quad -\infty<x<\infty,
	\end{align*}
	where $\phi$ denotes the PDF of the standard normal distribution.
	
	If we now choose $a = 0$, the classical slash (SL) PDF is obtained, given by
	\begin{align*}
	f_{\rm SL}(x;q)
	&=
	q \int_{0}^1 t^q \phi(t x)\, {\rm d}t,
	\quad -\infty<x<\infty
	\\[0,2cm]
	&=
	{q\, 2^{{q\over 2}-1}\over \sqrt{\pi}}\,
	\vert x\vert^{-(q+1)}\,
	\gamma\left({q+1\over 2},{x^2\over 2}\right).
	\end{align*}
	In this case, we denote it by $X\sim {\rm SL}(q)$.
\end{definition}

\begin{proposition}[Slash generator]\label{cond-marg-slash}
	Let  $\boldsymbol{Z}=(Z_1,Z_2)^\top$ be a random vector as in  Proposition \ref{Stochastic Representation}.
	If 
$g_c(x)= x^{-{(q+2)/ 2}} \gamma({(q+2)/ 2},{x/ 2})$,
then
the conditional distribution of $Z_{2}$, given $Z_{1}=x$, is ${\rm ESL}(x,q+1)$ and both of its unconditional distributions are ${\rm SL}(q)$.
\end{proposition}	
\begin{proof}
By \eqref{pdf-Zs} and using the definition of  $Z_{g_c}$ in Table \ref{table:1}, the joint PDF of $Z_1$ and $Z_2$ is given by
\begin{align}\label{eq-app3}
f_{Z_1,Z_2}(x,y)
&=
{q \, 2^{{q\over 2}-1}\over \pi}\, 
(x^2+y^2)^{-{q+2\over 2}}\, \gamma\left({q+2\over 2},{x^2+y^2\over 2}\right)
\nonumber
\\[0,2cm]
&=
q \int_{0}^1 t^{q+1} \phi(t x)\phi(t y)\, {\rm d}t.
\end{align}
So, the marginal PDF of $Z_1$ is
\begin{align}\label{eq-app4}
f_{Z_1}(x)
=
\int_{-\infty}^{\infty}
f_{Z_1,Z_2}(x,y) \,{\rm d}y
&=
q \int_{0}^1 t^{q+1} \phi(t x)\left[ \int_{-\infty}^{\infty}\phi(t y)\,{\rm d}y\right] {\rm d}t
\nonumber
\\[0,2cm]
&=
q \int_{0}^1 t^q \phi(t x)\, {\rm d}t
=
	f_{\rm SL}(x;q),
\end{align}
which proves that (see Definition \ref{def-GH-1}), $Z_1\sim  {\rm SL}(q)$. Similarly, we can show that $Z_2\sim  {\rm SL}(q)$.

On the other hand, by \eqref{eq-app3} and \eqref{eq-app4}, the conditional PDF of $Z_2$, given $Z_1=x$, is obtained as
\begin{align*}
	f_{Z_2\,\vert\,Z_1}(y\,\vert\, x)
	=
	\dfrac{\displaystyle\int_{0}^1 t^{q+1} \phi(t x)\phi(t y)\, {\rm d}t}{\displaystyle\int_{0}^1 u^q \phi(u x)\, {\rm d}u}
	=
	f_{\rm ESL}(y;x,q+1).
\end{align*}
From Definition \ref{def-GH-1}, we then find that $Z_2\,\vert\,(Z_1=x)\sim{\rm ESL}(x,q+1)$.
\end{proof}	

\section{Data sets}

\begin{table}[ht]
\small
\centering
\caption{UEFA Champions League and 2022 FIFA World Cup data sets.}
\begin{tabular}{rrrrrrrrrrr}
  \hline
  &\multicolumn{2}{c}{UEFA} & & &\multicolumn{2}{c}{FIFA}\\
  & W1 & W2 & & &  W1 & W2 \\
  \hline
  1 & 0.289 & 0.222 & & & 0.888 & 0.541 \\
  2 & 0.700 & 0.200 & & & 0.815 & 0.474 \\
  3 & 0.211 & 0.211 & & & 0.907 & 0.624 \\
  4 & 0.733 & 0.944 & & & 0.891 & 0.606 \\
  5 & 0.444 & 0.444 & & & 0.827 & 0.517 \\
  6 & 0.544 & 0.544 & & & 0.898 & 0.557 \\
  7 & 0.089 & 0.089 & & & 0.856 & 0.462 \\
  8 & 0.767 & 0.789 & & & 0.861 & 0.618 \\
  9 & 0.433 & 0.433 & & & 0.890 & 0.603 \\
  10 & 0.911 & 0.533 & & & 0.860 & 0.477 \\
  11 & 0.800 & 0.800 & & & 0.920 & 0.646 \\
  12 & 0.733 & 0.689 & & & 0.894 & 0.587 \\
  13 & 0.278 & 0.100 & & & 0.913 & 0.648 \\
  14 & 0.456 & 0.033 & & & 0.849 & 0.471 \\
  15 & 0.178 & 0.833 & & & 0.781 & 0.427 \\
  16 & 0.200 & 0.200 & & & 0.828 & 0.442 \\
  17 & 0.244 & 0.156 & & & 0.864 & 0.581 \\
  18 & 0.467 & 0.467 & & & 0.820 & 0.527 \\
  19 & 0.022 & 0.022 & & & 0.846 & 0.526 \\
  20 & 0.400 & 0.578 & & & 0.879 & 0.601 \\
  21 & 0.378 & 0.378 & & & 0.860 & 0.481 \\
  22 & 0.589 & 0.433 & & & 0.885 & 0.616 \\
  23 & 0.600 & 0.078 & & & 0.862 & 0.592 \\
  24 & 0.567 & 0.311 & & & 0.769 & 0.463 \\
  25 & 0.844 & 0.711 & & & 0.845 & 0.495 \\
  26 & 0.711 & 0.167 & & & 0.846 & 0.489 \\
  27 & 0.289 & 0.533 & & & 0.931 & 0.751 \\
  28 & 0.178 & 0.178 & & & 0.863 & 0.555 \\
  29 & 0.489 & 0.144 & & & 0.856 & 0.447 \\
  30 & 0.278 & 0.156 & & & 0.879 & 0.569 \\
  31 & 0.611 & 0.122 & & & 0.812 & 0.613 \\
  32 & 0.544 & 0.544 & & & 0.841 & 0.594 \\
  33 & 0.267 & 0.267 & & & --    & -- \\
  34 & 0.489 & 0.333 & & & --    & -- \\
  35 & 0.467 & 0.033 & & & --    & -- \\
  36 & 0.300 & 0.522 & & & --    & -- \\
  37 & 0.311 & 0.311 & & & --    & -- \\
   \hline
\end{tabular}
\label{table:datasets}
\end{table}

\end{appendices}
\end{document}